\documentclass[letterpaper,twocolumn,10pt]{article}
\usepackage{usenix2019,epsfig,endnotes}

\usepackage{tikz}
\usepackage{rotating}
\usepackage{graphicx}
\usepackage{amssymb}
\usepackage{amsfonts}
\usepackage{amsmath}

\newcommand{\pgftextcircled}[1]{
	\setbox0=\hbox{#1}%
	\dimen0\wd0%
	\divide\dimen0 by 2%
	\begin{tikzpicture}[baseline=(a.base)]%
	\useasboundingbox (-\the\dimen0,0pt) rectangle (\the\dimen0,1pt);
	\node[circle,draw,outer sep=0pt,inner sep=0.1ex] (a) {#1};
	\end{tikzpicture}
}

\let\textcircled=\pgftextcircled

\usepackage{xspace}
\usepackage[hyphens,obeyspaces]{url}
\usepackage{xcolor,colortbl}
\usepackage{pifont}
\usepackage{diagbox}
\usepackage{adjustbox}
\usepackage{pgfplots}
\pgfplotsset{width=8.5cm,compat=1.9}
\usetikzlibrary{patterns}
\usetikzlibrary{positioning}
\usepackage{tabularx}
\usepackage{hhline,tabu}
\usepackage{multirow}
\usepackage{listings}
\usepackage{booktabs}
\usepackage{dirtytalk}

\usepackage{xcolor}
\usepackage{pgfplots}
\usepackage{tikz}
\usepackage[inline]{enumitem}
\usepackage{float}
\usepackage{bytefield}

\definecolor{bblue}{HTML}{4F81BD}
\definecolor{rred}{HTML}{C0504D}
\definecolor{ggreen}{HTML}{9BBB59}
\definecolor{ppurple}{HTML}{9F4C7C}

\definecolor{barcolor1}{HTML}{A6CEE3}
\definecolor{barcolor2}{HTML}{1F78B4}
\definecolor{barcolor3}{HTML}{9BBB59}

\newcommand{\fooArg}{} 
\newcommand{\foooArg}{} 
\floatstyle{boxed}
\newfloat{Scheme}{t}{ext}
\newenvironment{scheme}[2]
{
\renewcommand{\fooArg}{#1}
\renewcommand{\foooArg}{#2}
\begin{Scheme}[t]
\hspace{-0.2cm}
\small
}{
\caption{\fooArg}
\label{\foooArg}
\end{Scheme}
}

\newcommand{\ourname}{\textit{simTPM}\xspace}
\newcommand{\ournameplain}{simTPM\xspace}

\newcommand{\tpm}{TPM\xspace}
\newcommand{\stpm}{simTPM\xspace}
\newcommand{\ftpm}{fTPM\xspace}
\newcommand{\atpm}{piTPM\xspace}
\newcommand{\etpm}{embTPM\xspace}
\newcommand{\boardid}{board\_id\xspace}
\newcommand{\cvtpm}{SCoP-vTPM\xspace}
\newcommand{\hvtpm}{SW-only-vTPM\xspace}

\newcommand{\cros}{\textbf{\textcolor{red}{\ding{56}}}\xspace}
\newcommand{\scros}{\textbf{\textcolor{red}{\ding{56}}}\xspace}

\newcommand{\tick}{\textbf{\textcolor{blue}{\ding{52}}}\xspace}
\newcommand{\1}{\textbf{\textcolor{black}{\ding{63}}}\xspace}

\newcommand{\no}{\textbf{\textcolor{red}{\ding{56}}}\xspace}
\newcommand{\kinda}{\textbf{\textcolor{orange}{\ding{74}}}\xspace}
\newcommand{\na}{\textbf{\textcolor{black}{\ding{82}}}\xspace}
\newcommand{\yes}{\textbf{\textcolor{blue}{\ding{52}\xspace}}\xspace}

\definecolor{Gray}{gray}{0.90}
\newcolumntype{a}{>{\columncolor{Gray}}c}

\newcommand{\circledb}[1]{
	\tikz[baseline=(myanchor.base)]
	\node[circle,fill=black,inner sep=1pt] (myanchor) {
		\color{-.}
		\bfseries
		\sffamily
		\footnotesize #1
	};
}

\definecolor{lightgray}{gray}{0.93}

\DeclareTextFontCommand{\red}{\color{red}}
\usetikzlibrary{trees,shapes.geometric,shapes.symbols,chains,calc,positioning,arrows,decorations.pathmorphing,decorations.pathreplacing}

\newlength\bubblesize
\definecolor{KWColor}{rgb}{0.37,0.08,0.25}
\definecolor{CommentColor}{rgb}{0.133,0.545,0.133}
\definecolor{StringColor}{rgb}{0,0.126,0.941}
\definecolor{Gray}{gray}{0.9}
\definecolor{lightgray}{gray}{0.93}
\newcommand{\rexec}{\leftarrow^{\hspace{-0.53em}\scalebox{0.5}{\$}}\hspace{0.2em}}
\newcommand{\GG}{\mathbb{G}}
\newcommand{\ZZ}{\mathbb{Z}}

\newcommand{\Enc}{\mathsf{Enc}}
\newcommand{\pp}{{\sf pp}}
\newcommand{\gpk}{{\sf gpk}}
\newcommand{\gsk}{{\sf gsk}}
 
\newcommand{\nym}{{\sf nym}}
\newcommand{\hash}{\mathsf{H}}

\newcommand{\Adv}{\mathsf{Adv}}
\newcommand{\bsn}{\mathsf{bsn}}
\newcommand{\A}{\mathcal{A}}
\newcommand{\exec}{\ensuremath{\leftarrow}}

\newcommand{\sk}{{\sf sk}}
\newcommand{\iss}{{\sf iss}}
\newcommand{\NP}{{\sf NP}}

\newcommand{\rand}{\rexec}
\newcommand{\crs}{{\sf crs}}

\newcommand{\extract}{{\sf Extract}}
\newcommand{\accept}{{\sf accept}}
\newcommand{\reject}{{\sf reject}}
\newcommand{\enc}{\Enc}
\newcommand{\keygen}{{\sf KeyGen}}
\newcommand{\E}{{\mathcal E}}
\newcommand{\skE}{{\sf sk_{\E}}}
\newcommand{\pkE}{{\sf pk_{\E}}}
\newcommand{\dec}{{\sf Dec}}
\newcommand{\add}{{\sf Add}}
\newcommand{\st}{{\sf st}}
\newcommand{\mulC}{{\sf MulC}}
\newcommand{\M}{{\mathcal M}}

\newcommand{\C}{{\mathcal C}}
\newcommand{\indcpa}{{\sf IND\text{-}CPA}}

\newcommand{\setupSoK}{{\sf Setup_{SoK}}}
\newcommand{\signSoK}{{\sf Sign_{SoK}}}
\newcommand{\verifySoK}{{\sf Verify_{SoK}}}
\newcommand{\simsetupSoK}{{\sf SimSetup_{SoK}}}
\newcommand{\simulSoK}{{\sf Sim_{SoK}}}
\newcommand{\simsignSoK}{{\sf SimSign_{SoK}}}
\newcommand{\extractSoK}{{\sf Extract_{SoK}}}

\usepackage{lipsum}                     
\usepackage{xargs}
\usepackage{todonotes}                     
\newtheorem{definition}{Definition}
\newtheorem{theorem}{Theorem}
\newtheorem{proof}{Proof}
\newtheorem{remark}{Remark}

\begin{document}

\date{}

\title{\Large \bf simTPM: User-centric TPM for Mobile Devices \\ (Technical Report)}

\author{
{\rm Dhiman Chakraborty}\\
CISPA Helmholtz Center \\ for Information Security, \\ Saarland University
\and
{\rm Lucjan Hanzlik}\\
CISPA Helmholtz Center \\ for Information Security, \\ Stanford University
\and
{\rm Sven Bugiel}\\
CISPA Helmholtz Center \\ for Information Security
} 

\maketitle

\thispagestyle{empty}

\subsection*{Abstract}

Trusted Platform Modules are valuable building blocks for security solutions and have also been recognized as beneficial for security on mobile platforms, like smartphones and tablets. However, strict space, cost, and power constraints of mobile devices prohibit an implementation as dedicated on-board chip and the incumbent implementations are software TPMs protected by Trusted Execution Environments.

In this paper, we present \ournameplain, an alternative implementation of a mobile TPM based on the SIM card available in mobile platforms. We solve the technical challenge of implementing a TPM2.0 in the resource-constrained SIM card environment and integrate our \stpm into the secure boot chain of the ARM Trusted Firmware on a HiKey960 reference board. Most notably, we address the challenge of how a removable TPM can be bound to the host device's root of trust for measurement. As such, our solution not only provides a mobile TPM that avoids additional hardware while using a dedicated, strongly protected environment, but also offers promising synergies with co-existing TEE-based TPMs. In particular, \ournameplain offers a user-centric trusted module. Using performance benchmarks, we show that our \stpm has competitive speed with a reported TEE-based TPM and a hardware-based \tpm.

\section{Introduction}

Trusted computing technology has become a valuable building block for security solutions.
The most widely deployed form of trusted computing on end-consumer devices is the Trusted Platform Module~(TPM), a dedicated hardware chip that offers facilities for crypto co-processing, protected credentials, secure storage, or even the attestation of its host platform's state. By today, software and system vendors have built various security solutions on top of TPM. For instance, Microsoft's BitLocker uses it to release disk-encryption credentials only to a trustworthy bootloader~\cite{windows_tpm}; or Google's Chromium uses the TPM for a range of objectives~\cite{tpm_chromium}, such as preventing software version rollback, protecting RSA keys, or attesting protected keys.

TPM is also of interest for the different stakeholders on mobile devices. However, the particular benefits that the TPM offers have historically hung on the TPM's implementation as a dedicated security chip that can act as a "local trusted third party" on devices. Mobile devices are, however, constrained in space, cost, and power consumption, which prohibits a classical deployment of TPM. To address the particular problems of the mobile domain, 
the Trusted Computing Group~(TCG) introduced the Mobile Trusted Module~(MTM) specifications~\cite{TCGMTMSpe}. Although the MTM concept has never left the prototype status, its ideas influenced the latest TPM2.0 specification~\cite{TPM2Spec}. The TPM2.0 mobile reference architecture~\cite{TCGMTMSpec2} proposed different alternatives for implementing a TPM on a mobile device, including virtualization, dedicated cores, or hardware-based isolation. The de-facto implementation of mobile TPMs today are protected environments through hardware-based trusted execution environment~(TEE)~\cite{RajSWACEFKLMNRS16,McGillionDNA15,JangKKKK15,EkbergKA14,Ekberg13,McCunePPRI08}, like ARM TrustZone that is available on virtually all mobile platforms today, where the TPM is implemented as protected software application inside the TEE.

Given the different proposals for realizing TPMs on mobile platforms, we conduct a systematic comparison of the different  solutions in terms of security of the TPM itself, their applicability in current systems, and deploy-ability in the specific setting of mobile devices. While the solutions naturally differ in their security guarantees for the TPM (i.e., TPM state or execution) due to differences in the underlying technology (e.g., dedicated hardware chip vs.~virtual machine), we see particularly shortcomings of the current solutions in terms of applicability and deploy-ability. In particular, the currently incumbent \ftpm (firmware TPM) is strictly bound to the platform vendors and serves their purposes (e.g., securing vendor credentials), but is not or only very limited available to other stakeholders in the system, such as the user. Moreover, an \ftpm~\cite{RajSWACEFKLMNRS16} that is based on a TEE falls short on providing a fully measured boot by itself. The availability of an \ftpm depends on the availability of the TEE during boot, which is one of the last steps in the long boot-chain. In light of recent attacks against mobile bootloaders~\cite{bootstomp} and trusted software in TEE~\cite{MachiryGSSSWBCK17,TZDowngrade,ExploitingTZTEE,bitsPlease,BHTZ,BHRose,nailgun:sp19,197209,203864}, this lacking support to attest the entire, early boot-chain, including the software in the TEE, is unsatisfactory.

To put a new perspective on solving those issues of \ftpm,
we add in this paper an alternative implementation of a hardware TPM called \ourname to the landscape of mobile TPM implementations by using the subscriber identity module~(SIM) card. We have implemented a prototype of our solution on a Hikey960 reference board~\cite{hikey960} and using a Gemalto Multos card as SIM card.
Our \ournameplain solves the technical challenge of implementing TPM2.0 compliant functionality on the SIM card, which does not require any additional hardware for the TPM. This approach keeps the costs down and
leverages dormant hardware capabilities of mobile devices. Through performance tests, we show that \ournameplain is competitively fast to reported \ftpm implementations. A particular challenge of this design is the lack of the usual physical binding between the TPM and its host platform's root of trust for measurement~(RTM), that is, a SIM card can be moved to another platform. We discuss two strategies in the particular setting of mobile devices on how to bind the \ournameplain to a device's RTM, either through an extended secure boot and TEE proxy or through a distance bounding protocol. Once bound to the device's RTM, we also integrated \ournameplain with the ARM Trusted Firmware~(ATF) boot chain to augment the ATF secure boot with an authenticated boot. Our solution not only fills the gap of TEE-based TPMs for measured boots, but the co-existence of a \ftpm and \stpm on a mobile device creates also promising synergies between the two TPMs (e.g., to support multiple stakeholders).
Our contribution can be summarized as follows:
\begin{enumerate}[topsep=0pt,parsep=0pt,partopsep=0pt]
	\item A systematic comparison of existing solutions for mobile TPMs and their enabling technologies. We discover that incumbent solutions fall short on applicability and deploy-ability aspects.

	\item We implemented the first SIM card based TPM2.0 for mobile devices by developing a \stpm, which can be executed in this constrained environment. Our solution enables a user-centric trusted module offering a portable sealed storage.

	\item We propose an integration with the on-board TEE to solve the problem of binding the \stpm to the RTM and discuss an alternative solution based on distance bounding. As a result of this binding, a fully measured boot on the ARM Trusted Firmware~(ATF) secure boot chain is possible.

	\item The performance of our \stpm is competitively fast to a reported \ftpm implementation and is comparable with existing hardware TPMs.
\end{enumerate}

\section{Background}
\label{sec:background}

We briefly introduce necessary background information about ARM Trusted Firmware, TPM, and SIM cards. 

\subsection{ARM Trusted Firmware (ATF)}
\label{background:atf}
ATF implements a subset of the trusted board boot requirements for ARM reference platform~\cite{TBBR}. Figure~\ref{fig:secureBootProcess} illustrates the bootloader settings and boot chain. ATF is triggered when the platform is powered on. After the primary CPU and all other CPU cores are initialized successfully, the primary core triggers the ATF (\textcircled{P1}). ATF is divided in five steps depending on modularity: \circledb{1}~BootLoader stage~1~(BL1) for AP trusted boot ROM, \circledb{2}~BootLoader stage~2~(BL2) for Trusted Boot Firmware, \circledb{3}~BootLoader stage~3-1~(BL3-1) for EL3 Runtime Firmware, \circledb{4}~BootLoader stage~3-2~(BL3-2) for Secure-EL1 Payload (optional), \circledb{5}~BootLoader stage 3-3~(BL3-3) for Non-trusted Firmware.

\textbf{Secure boot:} ATF implements a secure boot in which every component along the boot chain \textcircled{P\#} verifies the authenticity and integrity of the next component. Since BL1 does not have a preceding component, it has to be axiomatically trusted. Thus, BL1 verifies BL2, BL2 verifies BL3.x, and so forth. Verification is usually based on certificates, where a hash of a trusted (vendor) public key is fused into the hardware and is available to BL1 to ensure a trustworthy signature of BL2. At the end of a successful secure boot, every component in the boot chain has been checked for integrity and authenticity before handing control to it. If any verification fails, the boot aborts.

\begin{figure}
	\centering
	\includegraphics[width=0.370\textwidth]{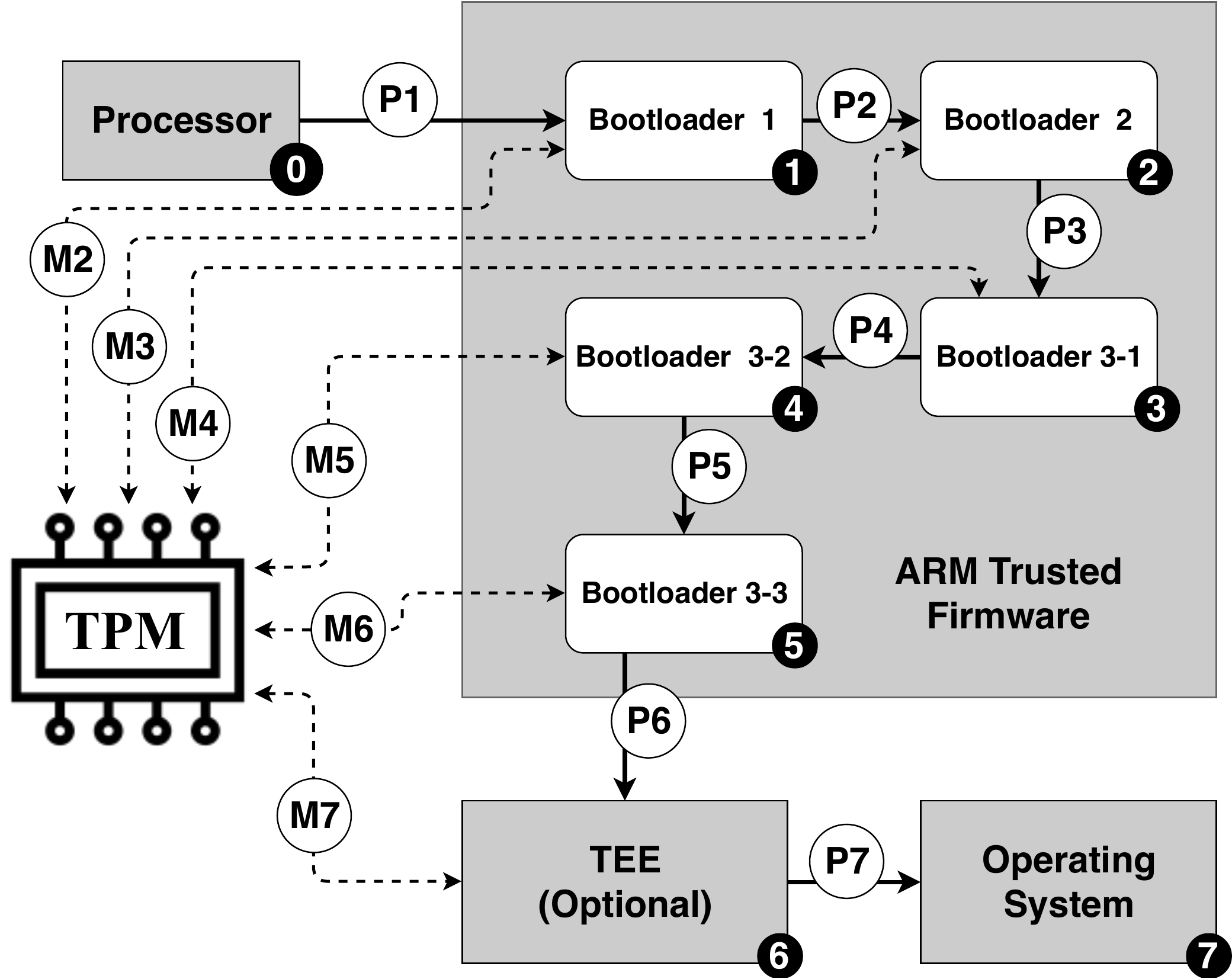}
	\caption{Trusted Boot Process with TPM; P(\#) = boot chain path; M(\#) = measurement of component \#}
	\label{fig:secureBootProcess}
\end{figure}

\subsection{Trusted Platform Module (TPM)}
\label{background:tpm}
TPM by the Trusted Computing Group is the most wide-spread trusted computing technology on end-user devices. By today, the \tpm specification is in its version 2.0, addressing many of the security issues and practical concerns of previous versions 1.0--1.2. According to this specification, a \tpm provides a number of desirable hardware and security features. It is equipped with secure non-volatile memory, a set of platform configuration register (PCR) banks, a processor to run \tpm code in isolation, co-processors for common cryptographic primitives (e.g., RSA, ECC, SHA-1, SHA-256), a clock, and a random number generator. By default, a \tpm is deployed as a hardware chip soldered onto a platform's motherboard. Besides acting as a cryptographic co-processor, a \tpm provides the facilities to securely store measurement about the host platform's configuration (e.g., software state) in its PCRs and to reliably report those measurements to a remote verifier (remote attestation based on a pre-installed endorsement key), as well as creating secure storage through \tpm protected credentials and data sealing with extended authorization policies. Further, the \tpm non-volatile memory, including secure monotonic counters, can be attractive for building security solutions, e.g., version rollback prevention for software updates.

By now, a number of real world applications make use of \tpm. For instance, IBM's password manager uses it for storing keys, Microsoft windows management instrumentation uses TPM for cryptographic co-processing, Intel's Trusted eXecution Technology or AMD's Secure Technology rely on a hardware TPM, several VPN apps can make use of it, \tpm is used in full disk encryption (e.g., Microsoft Bitlocker, dm-crypt), and even browsers like Chrome make use of \tpm for different purposes.

\textbf{Measured boot:} Of particular relevance for this paper is measured (or authenticated) boot based on \tpm (see Figure~\ref{fig:secureBootProcess}). During a measured boot, every component in the boot chain \textcircled{P\#} measures the next component---a cryptographic hash of the component---and then stores this measurement in the PCR of the \tpm (\textcircled{M\#}) before passing on control. Since BL1 does not have a preceding component, it is not measured and acts as the \emph{Root of Trust for Measurement}, which starts the measurement chain. In contrast to a secure boot, the components are not verified and the boot is not aborted, however, after a measured boot the software configuration of the boot components can be attested by the \tpm or used to seal storage to this configuration (i.e., values in the PCR).

\subsection{Subscriber Identification Module (SIM)}
\label{background:sim}
SIM card is the module that authenticates the mobile device in the network. The primary job of the SIM card is to prove the identity of the owner of subscription to the cellular carrier to enable services like calling, Internet, and various others.

Through physically separated pins, a SIM module can achieve the same degree of independence from power supply, reset capability, clock signal, and separated I/O communication with the host platform like a TPM.

Since SIM cards are smart cards, they use command-response communication and the application protocol data unit (APDU) to communicate with their reader. The Android radio interface layer can be extended to send specialized APDU commands to the SIM card, which we use in \ournameplain. It is worth noting, that this APDU command sent by the Android radio interface has to go through the baseband processor. The structure of the APDU commands are defined in the ISO/IEC 7816-4 standard and are recalled later on in Section~\ref{sec:systemdesign}.

\subsection{Cryptographic Preliminaries}

We will use $y \exec \A(x)$ to denote 
the execution of algorithm $\A$ outputting $y$, on input $x$. 
By $r \rexec S$ we 
mean that $r$ is chosen uniformly 
at random over the set $S$.
In our construction we will make use of elliptic curve groups where 
we will use $1_{\GG}$ to
denote the identity element in group $\GG$
and $[n]$ to denote the set $\{1,\ldots,n\}$.
Moreover, we will make use of bilinear maps defined as follows:

\begin{definition}[Bilinear map]
Let us consider  
cyclic groups
 $\mathbb{G}_1$, $\mathbb{G}_2$, $\mathbb{G}_T$
 of a
prime order $p$.
Let $g_1, g_2$ be generators of respectively
$\mathbb{G}_1$ and $\mathbb{G}_2$.
We call $e: \GG_1 \times \GG_2 \rightarrow \GG_T$
a \emph{bilinear map} (pairing) if it is efficiently computable and the following
holds:
\begin{description}
\item[Bilinearity:] $\forall (S,T) \in \mathbb{G}_1 \times \mathbb{G}_2$,
$\forall a, b \in \ZZ_p$, we have $e(S^a,T^b) = e(S,T)^{a\cdot b}$,
\item[Non-degeneracy:] $e(P_1,P_2)\not=1$ is a generator of group $\GG_T$.
\end{description}
\end{definition}

\paragraph{Signature of Knowledge}
We now recall the definition of signatures of knowledge (SoK) formalized by Chase and Lysyanskaya \cite{SoK}. 
In a signature scheme by signing a message we show that the knowledge of the
corresponding secret key. For SoK schemes the idea is a bit different.
We are given a $\NP$ language $L$  and a statement $x \in L$, associated with the user's identity. By signing
a message the user shows that he knows the hard-to-find witness $w$ for this statements. More formally,
a signature of knowledge scheme is defined as follows.
\begin{definition}[SoK]
Let $L$ be a $\NP$ language defined by a polynomial time Turing machine $M_{L}$, such
that all witnesses for $x \in L$ are of known polynomial length $p(\lvert x \rvert)$. Then 
$(\setupSoK,\signSoK,\verifySoK)$ is a signature of knowledge of a witness for $L$, for a message space $\M$ if
the following properties hold:
\begin{description}

\item[Correctness] For all $\lambda$, $x \in L$, valid witnesses $w$ for $x$ (i.e. such that $M_L(x,w)=1$), and
$m \in \M$ the probability
\begin{align*}
1 - \Pr[\crs \exec& \setupSoK(\lambda), \\
\sigma \exec& \signSoK(\crs,M_L,x,w,m):& \\
\accept \exec &\verifySoK(\crs,M_L,x,m,\sigma)  ]
\end{align*}
is a negligible function of $\lambda$.

\item[Simulatability] There exists a PPT simulator $(\simsetupSoK,\allowbreak \simsignSoK)$ such that
for all PPT adversaries $\A$ there exists a negligible function $\epsilon$ such that 
for all polynomials $f$, for all $\lambda$, for all 
auxiliary inputs $s \in \{0,1\}^{f(\lambda)}$
\begin{displaymath}
 \begin{vmatrix}
\Pr[(\crs,\tau) \exec \simsetupSoK(\lambda), \\
b \exec \A^{\simulSoK(\crs,\tau,\cdot,\cdot,\cdot,\cdot)}(s,\crs): b=1] \\
- \Pr[\crs \exec \setupSoK(\lambda), \\
 b \exec \A^{\signSoK(\crs,\cdot,\cdot,\cdot,\cdot)}(s,\crs): b=1] 
\end{vmatrix} = \epsilon(\lambda)
\end{displaymath}
where the oracle $\simulSoK$ receives the values $(M_L,x,w,m)$ as inputs, checks that
the witness $w$ given to it was correct and returns $\sigma \exec \simsignSoK(\crs,\tau,M_L,x,m)$. 
$\tau$ is a trapdoor used by the simulator to simulate signatures without knowing a witness.
\begin{remark}
In schemes relying on the random oracle model, this trapdoor is usually not used.
\end{remark}

\item[Extraction] There exists an extractor algorithm $\extract$ such that for all PPT adversaries
$\A$ there exists a negligible function $\epsilon$ such that 
for all polynomials $f$, for all $\lambda$, for all 
auxiliary inputs $s \in \{0,1\}^{f(\lambda)}$
\begin{align*}
\Pr[&(\crs,\tau) \exec \simsetupSoK(\lambda),  \\
&(M_L,x,m,\sigma) \exec \A^{\simulSoK(\crs,\tau,\cdot,\cdot,\cdot,\cdot)}(s,\crs)  \\
&w \exec \extractSoK(\crs,\tau,M_L,x,m,\sigma): M_L(x,w) =1  \\
&\lor (M_L,x,m,w) \in Q \\
&\lor \verifySoK(\crs,M_L,x,m,\sigma) = \reject] = 1 - \epsilon(\lambda),
\end{align*}
where $Q$ denotes a list of all previous queries $(M_L,x,m,w)$ $\A$ has sent to the oracle $\simulSoK$.

\end{description}
\end{definition}

Typically we obtain signatures of knowledge by converting a three move
zero-knowledge 
proof of knowledge (e.g.  a $\Sigma$-protocol) to 
a signature protocol
by applying the
Fiat-Shamir heuristic \cite{DBLP:conf/crypto/FiatS86}.
We can for example transform a $\Sigma$-protocol for proving the knowledge of
a discrete logarithm of $X$ to a signature of knowledge. We will use the Camenisch-Stadler 
notation \cite{CSNotation} to describe signature of knowledge.
So, the above mentioned $\Sigma$-protocol transformed into
a signature of knowledge on message $m$ will be denoted as
$$
SoK\{(\alpha): g^{\alpha} = X \}(m),
$$
where in this case $\alpha$ is the witness which knowledge we show and $X$ is part of the statement.

\paragraph{Additively Homomorphic Encryption Scheme}
Homomorphic encryption schemes are standard cryptosystems that allow for an
additional operation on ciphertexts. In case of additively homomorphic schemes,
given ciphertext of message $m_1$
and message $m_2$ we can compute the ciphertext of message $m_1+m_2$ (for most 
schemes this is modular addition)
without the knowledge of $m_1, m_2$ and the secret key. 
More formally, we define a
additively homomorphic encryption scheme as follows.
Note that to simplify notation 
we added two algorithms $\add$ and $\mulC$. 

\begin{definition}
An additively homomorphic encryption scheme $\E$ consists of the
following PPT algorithms $(\keygen,\enc,\dec,\add,\mulC)$:
\begin{description}
\item[$\keygen(\lambda)$:] on input security parameter $\lambda$, this
algorithm outputs a secret key $\skE$ and public key $\pkE$. 

\item[$\enc(\pkE,m)$:] on input message $m \in \M$ and public key $\pkE$, this
algorithm outputs ciphertext $c \in \C$. 

\item[$\dec(\skE,c)$:] on input ciphertext $c \in \C$ and secret key $\skE$, this
algorithm outputs message $m \in \M$. 

\item[$\add(\pkE,c_1,c_2)$:] on input ciphertext $c_1 \in \C$, ciphertext $c_2 \in \C$ and public key $\pkE$, this
algorithm outputs ciphertext $c_3 \in \C$. 

\item[$\mulC(\pkE, c_1,m)$:] on input ciphertext $c_1 \in \C$, message $m \in \M$ and public key $\pkE$, this
algorithm outputs ciphertext $c_2 \in \C$. 
\end{description}

We require that an additively homomorphic encryption scheme is correct and $\indcpa$ secure, which
we define as follows:
\begin{description}

\item[Correctness] For all $\lambda$, all $(\skE,\pkE) \exec \keygen(\lambda)$, all $m \in \M$,
all pairs $(m_1,m_2) \in \M \times \M$ such that $m=m_1+m_2$, all pairs $(m_3,m_4) \in \M \times \M$ 
such that $m = m_3 \cdot m_4$, we have
$\dec(\skE,\enc(\pkE,m)) = m$.
Additionally, we have
$\dec(\skE,c_{1,2}) = m$,
where $c_1 \exec \enc(\pkE,m_1)$, $c_2 \exec \enc(\pkE,m_2)$, $c_{1,2} \exec \add(\pkE, c_1,c_2)$
and $\dec(\skE, c_{3,4}) =m$
where $c_3 \exec \enc(\pkE, m_3)$,
and $c_{3,4} \exec \mulC(\pkE, c_3,m_4)$.\\

\item[Indistinguishability under chosen-plaintext attack] An encryption scheme $\E$ is $\indcpa$ secure, if for all PPT adversaries $\A$ we have that the advantage of the adversary $\Adv_{\E,\A}^{\indcpa}(\lambda)$ defined as the probability: 
\begin{align*}
 \lvert \Pr&[(\skE,\pkE) \exec \keygen(\lambda), (m_0,m_1,\st) \exec \A(\pkE), \\
& b \rexec \{0,1\},
c \exec \enc(\pkE,m_b), \hat{b} \exec \A(\st,c): b = \hat{b} ],
\end{align*}
where we require that $m_0,m_1 \in \M$, is a negligible funtion of $\lambda$.

\end{description}
\end{definition}

Notable additively homomorphic encryption schemes are the Paillier cryptosystem \cite{Paillier} and the Okamoto-Uchiyama cryptosystem \cite{OkamotoUchiyama}. Paillier's scheme uses a large RSA modulus $n = p \cdot q$ and operations in the group $\ZZ_{n^2}^*$, where also the ciphertext lies. The message space is $\ZZ_n$.
The Okamoto-Uchiyama cryptosystem uses a modulus of the form $n = p^2 \cdot q$, the ciphertext lies in $\ZZ_n^*$ and the message space is $\ZZ_p$.

\section{Requirement Analysis \& Systematization of Existing Solutions }
\label{sec:reqanalysis}
There exists many approaches to realize \tpm in a way different than using a dedicated hardware TPM.
In this section, we systematically compare different solutions of trusted computing procedures using both hardware and software that are representative for the different implementation options. For comparison, we first re-enumerate the objectives a secure and practical \tpm implementation needs to fulfill (Section~\ref{relatedWork:Objective}) and then discuss the existing solutions (Sections~\ref{relatedWork:fTPM} through~\ref{relatedwork:sgx}). In particular, this systematization should help to understand the trade-off of the proposed solutions in comparison to the default hardware \tpm and where our \ournameplain solution fits into. Table~\ref{tab:countermeasures} summarizes the discussion in the remainder of this section.

\subsection{Objectives}
\label{relatedWork:Objective}
We start by briefly formulating the objectives a trusted module, in particular for mobile devices, should fulfill. We group them into security of the \tpm 
itself, the applicability of the implementation, and desirable deploy-ability objectives.

\subsubsection{Security of \tpm} These are objectives that should be fulfilled to ensure the security of the \tpm state, its execution and trustworthiness, and secure operations.
\begin{description}[labelindent=0cm, leftmargin=0cm]
	\item [\textbf{S1}]\textbf{Confidentiality and integrity of \tpm state:} The \tpm state should be confidential and protected against untrusted code (e.g., host platform, non-TEE apps) and only be available to authorized entities. We assign \tick if the confidentiality and integrity of the state is protected through strong security means (e.g., physical isolation), \kinda if they depend on software integrity (e.g., of the OS), and \no in other cases.

	\item [\textbf{S2}]\textbf{Rollback Protection:} Reverting the \tpm state back to a former version must be prevented or at least be detectable. We assign \tick if rollback protection is guaranteed through hardware means (e.g., hardware counters), \kinda if there is a dependency on untrusted OS but rollbacks can be detected, \no if no rollback protection or detection is provided.

	\item [\textbf{S3}]\textbf{Trustworthy Endorsement:} A \tpm should be carrying an asymmetric encryption key called Endorsement key (EK) that can live as long as the \tpm and for which credentials exist that verify the authenticity of the \tpm and allow a verifier to recognize a genuine \tpm. We assign \tick if endorsement credentials are available to the \tpm (e.g., pre-installed at manufacturing time or derived from other verifiable credentials), \kinda if the \tpm has to create an EK and prove it is genuine through a remote verification, \no otherwise.

	\item [\textbf{S4}]\textbf{Secure Counter:} \tpm has to provide secure, persistent monotonic counters, e.g., for its clients or extended authorization policies. We assign \tick if the \tpm provides such counters backed by hardware support or NV-storage of the \tpm software state that is protected (i.e., \textbf{S1, S2} both \tick). We assign \kinda if the security of the counter depends on software integrity (e.g., of the OS or hypervisor). Otherwise \no.

	\item [\textbf{S5}]\textbf{Secure Clock:} A clock is needed for attestation, for generation of timed attestation keys, and for authorization policies with lock-out time. If a secure clock is available to the \tpm (e.g., its own hardware clock), we assign \tick; if the clock depends on shared resources but manipulation can be detected we assign \kinda, otherwise \no.

	\item [\textbf{S6}]\textbf{Security of \tpm Execution:} The execution of the \tpm code or firmware has to be protected against compromise. We assign \tick if a strong security boundary exists between untrusted code and the \tpm execution environment (e.g., dedicated physical chip). If the execution environment shares hardware resources (e.g., CPU or RAM) with untrusted code and the shared resources provide isolation (e.g., modes of operation of CPU and separate memory regions), we assign \kinda, since the shared resources open an attack surface. If the security of the \tpm execution environment is based purely on software means (e.g., hypervisor or OS), we assign \no for this weakest form of isolation.

\end{description}

\subsubsection{Applicability} These are objectives related to the application of \tpm, such as authenticated boot or providing secure storage to clients.
\begin{description}[labelindent=0cm, leftmargin=0cm]
	\item [\textbf{A1}]\textbf{Secure Persistent Storage:}  \tpm provides a persistent storage to securely store limited amounts of data (e.g., certificates). We assign \yes if the \tpm provides such storage (e.g., NV-RAM in a dedicated chip) and \kinda if the persistent storage is part of an outsourced \tpm state that is protected (i.e., \textbf{S1, S2} both \tick). We assign \scros in other cases. 

	\item [\textbf{A2}]\textbf{Early Availability:} A main use-case for \tpm is storing the measurement of loaded software components, i.e., measured boot. To be able to attest the entire software stack, the \tpm has to be early available during the boot sequence. If the trusted module is available as soon as the platform has power, we assign \tick. Otherwise, if the \tpm becomes available at late stage during boot (e.g., after initializing a separate execution environment), we assign \no.

	\item [\textbf{A3}]\textbf{Multiple Stakeholders:} Computer systems, in particular mobile platforms and enterprise devices, usually have multiple stakeholders co-existing with an interest in protecting credentials and software on the platform (e.g., end-user, administrator, network operator, software vendor). If the \tpm was designed to support both platform software and users (e.g., distinct hierarchies), we assign \yes. If the \tpm primarily supports the platform but offers limited functionality to the end-user, we give \kinda. If the \tpm was designed solely as support for the platform vendor, we give \no.
\end{description}

\subsubsection{Deploy-ability} Objectives related to the deployment of \tpm, in particular if deployment complies with the requirements of mobile devices or if it is bound to a specific platform.
\begin{description}[labelindent=0cm, leftmargin=0cm]
	\item[\textbf{D1}]\textbf{Mobile Availability:} We want to have the \tpm available for mobile devices. This imposes strict constraints, such as not changing the current architecture by adding a new on-board chip. If the \tpm implementation adheres to this constraints, we assign \tick, otherwise \no.

	\item[\textbf{D2}]\textbf{Movability:} The TCG specification has introduced the \tpm as being bound to its host platform (e.g., fixed part of the motherboard). However, depending on the context, the movability of the \tpm to another platform is desirable, e.g., if an associated virtual machine migrates to another platform. If the \tpm is generally easily moved to another platform, we assign \yes, if it is bound to a specific platform, we assign \no.

	\item[\textbf{D3}]\textbf{Bound RTM:} The measurements during a measured boot are given to the \tpm by the host platform, starting with the Root of Trust for Measurement~(RTM). To ensure that the provided measurements indeed describe the \tpm's host platform's configuration, \tpm and RTM must be bound together on the same platform. If this binding is achieved via physical means (e.g., \tpm and RTM are fixed parts of the same motherboard), we assign \yes. If the \tpm receives those measurements from another trusted entity (e.g., another, bound \tpm, or a secure boot anchored at the RTM), we assign \kinda. If the \tpm cannot establish trust into the RTM, we assign \no.
\end{description}

In Section~\ref{background:tpm}, while introducing the hardware \tpm, we explained all its properties, which allow the \tpm to achieve the objectives we defined in Section~\ref{relatedWork:Objective} and summarized in Table~\ref{tab:countermeasures}. Objectives \textbf{S1} to \textbf{S6} and \textbf{A1} to \textbf{A3} are our interpretation of properties derived from TCG's mobile TPM~\cite{TCGMTMSpe, TCGMTMSpec2} and standard TPM specification~\cite{TPM1Spec, TPM2Spec}. We define \textit{Deploy-ability} as added objectives that \ournameplain should achieve. The current TCG specifications do not stipulate a removable TPM.  We will use the standard hardware \tpm as the baseline that \ournameplain should achieve. 

\begin{table*}
\caption{Comparison of existing \tpm implementations}
\label{tab:countermeasures}
	\centering
	\begin{tabular}{@{} p{3.5cm} p{3.5cm} p{4.7cm}  *{3}{c >{\columncolor{Gray}}c}  }\\
		\tikz \node {\textbf{Category}}; &
		\tikz \node {\textbf{}}; &
		\tikz \node {\textbf{Objective}}; &
		\tikz \node[rotate=90] {\ftpm~\cite{RajSWACEFKLMNRS16}}; &
		\tikz \node[rotate=90] {vTPM$^{\1}$~\cite{BergerCGPSD06}}; &
		\tikz \node[rotate=90] {Intel SGX~\cite{journals/iacr/CostanD16}}; &
		\tikz \node[rotate=90] {\ournameplain}; &
		\tikz \node[rotate=90] {Hardware \tpm};
	    \\

		\hline
		\multirow{8}{\linewidth}{Security of \tpm} &  
		\multirow{3}{\linewidth}{Security of \tpm state} & 
		\textbf{S1.} Confidentiality and integrity  & 
		\yes & \yes / \kinda & \na & \yes & \yes 

		\\
		\multicolumn{1}{ c  }{}&
		\multicolumn{1}{ c  }{}&
		\textbf{S2.} Rollback protection  & 
		\yes & \yes / \kinda & \yes & \yes & \yes 

		\\
		\hhline{~*{7}{-}}
		\multicolumn{1}{ c  }{}&
		\multicolumn{1}{ c  }{}&
		\textbf{S3.} Trustworthy Endorsement  & 
		\yes & \kinda / \kinda & \yes & \yes &  \yes  

		\\
		\multicolumn{1}{ c  }{}&
		\multicolumn{1}{ c  }{}&
		\textbf{S4.} Secure counter  & 
		\yes & \yes / \kinda & \yes & \yes & \yes 

		\\
		\multicolumn{1}{ c  }{}&
		\multicolumn{1}{ c  }{}&
		\textbf{S5.} Secure clock  & 
		\kinda & \yes / \no & \no & \yes & \yes 

		\\
		\multicolumn{1}{ c  }{}&
		\multicolumn{1}{ c  }{}&
		\textbf{S6.} Security of \tpm execution & 
		\kinda & \yes / \no & \na & \yes & \yes 

		\\ \hline
		\multirow{3}{\linewidth}{Applicability} &
		\multirow{3}{\linewidth}{}&
		\textbf{A1.} Secure persistent storage  & 
		\kinda & \yes / \no & \kinda & \yes & \yes  

		\\
		\multicolumn{1}{ c  }{}&
		\multicolumn{1}{ c  }{}&
		\textbf{A2.} Early availability  & 
		\no & \yes / \yes & \no & \yes & \yes 

		\\
		\multicolumn{1}{ c  }{}&
		\multicolumn{1}{ c  }{}&
		\textbf{A3.} Multiple stake holder  & 
		\no & \yes / \yes & \yes & \yes & \yes 

		\\ \hline
		\multirow{2}{\linewidth}{Deploy-ability}&
		\multirow{2}{\linewidth}{}&
		\textbf{D1.} Mobile availability  & 
		\yes & \no / \no & \no & \yes & \no 

		\\
		\multicolumn{1}{ c  }{}&
		\multicolumn{1}{ c  }{}&
		\textbf{D2.} Movability  & 
		\scros & \tick / \tick & \scros & \tick & \scros 

		\\
		\multicolumn{1}{ c  }{}&
		\multicolumn{1}{ c  }{}&
		\textbf{D3.} Bound RTM  & 
		\tick & \no / \kinda & \tick & \kinda & \tick 

		\\\hline
		\multicolumn{8}{c}{\scriptsize \tick = fulfilled by the implementation; \kinda = partially fulfilled by the implementation; \scros = not fulfilled by the implementation; \na = not applicable for the implementation}\\
		\multicolumn{8}{c}{\scriptsize \1 First column is for \textit{Secure co-processor based vTPM} (SCoP) implementation and second column is for \textit{Software only vTPM} (SW-only) implementation}
	\end{tabular}
\end{table*}

\subsection{fTPM}
\label{relatedWork:fTPM}
Specifically for the mobile domain, a number of past implementations~\cite{Winter08,bugiel09:stc,RajSWACEFKLMNRS16} leveraged trusted execution environments~(TEE) to realize a software-based \tpm. We use Microsoft's \ftpm~\cite{RajSWACEFKLMNRS16} as a representative for those implementations, since it is one of the most recent solutions. 
The \ftpm implementation is widely deployed in Microsoft mobile devices using a TEE on top of ARM TrustZone (\textbf{D1:~\tick}). TrustZone creates a memory and process isolation between the protected environment ("secure world") running inside the TEE and the "normal world" (i.e., Android or similar), and allows the execution to switch contexts between those two worlds via a secure monitor. 

\ftpm provides confidentiality, integrity (\textbf{S1:~\tick}), and rollback protection (\textbf{S2:~\tick}) for \ftpm states by creating a trusted storage through a combination of encryption with fused keys, device UUID, and Replay Protected Memory Block~(RPMB) with authenticated writes and write counter. Any form of secure persistent storage the \ftpm offers to clients is based on this securely outsourced state (\textbf{A1:~\kinda}), which is also used to provide secure counters to clients (\textbf{S4:~\tick}).

Due to ARM TrustZone, the execution of the \ftpm environment is isolated from the normal world, however, both worlds still share the CPU and RAM (\textbf{S6:~\kinda}), which has opened TrustZone TEEs to attacks~(e.g.,~\cite{MachiryGSSSWBCK17}).

\ftpm does not have a separate secure clock. It uses the clock of the system in cooperation with the untrusted OS (\textbf{S5:~\kinda}). To handle the shared clock situation, \ftpm implements \textit{fate sharing}, where \ftpm refuses to provide any functionality if the OS does not cooperate. 

\ftpm is primarily designed to provide \tpm support to the platform vendor (\textbf{A3:~\scros}). The \ftpm is a software implementation and bound to one device (\textbf{D2:~\scros}), since it derives many of its credentials from device-specific keys or UUIDs, including its endorsement credentials (\textbf{S3:~\yes}).

Since the \ftpm is implemented as software in the TEE on top of ARM TrustZone, the \ftpm becomes only available once the TEE has been initialized during the boot sequence (see also Section~\ref{sec:background}). That means the \ftpm (or any TEE-based \tpm) is not early enough available to store measurements of the early boot stages (\textbf{A2:~\no}). But this can be alleviated by introducing shared memory between the bootloaders and TEE for measurement storage. We will discuss this solution in more details in Section~\ref{bindingToTheDevice}.

Although the \ftpm is only available after the bootchain has created the TEE, the secure boot transitively extends the trust put into the RTM (BL1) to the remainder of the secure bootchain on the same platform as the TEE. Thus, \ftpm can assume that the measurements are done as if by the RTM on the same platform (\textbf{D3:~\yes}) if the measurements comes from a component of the secure bootchain.

\subsection{vTPM}
Another way of implementing a software TPM is by creating virtual instances over a physical TPM~\cite{BergerCGPSD06}. This, in particular, targets cloud environments in which virtual machines need a \tpm, but sharing a single physical \tpm (or providing an array of physical \tpm) is not an option. The representative work for virtual TPM, or vTPM, is based on the Xen hypervisor and proposes two different implementation options: 1)~a software only implementation with vTPM instances running inside a privileged VM, and 2)~a secure co-processor (SCoP) to run all vTPM instances with better isolation at the cost of additional hardware. Both options are not feasible for mobile TPMs (\textbf{D1:~\no}), since virtualization is not sufficiently supported or effective, and adding a secure co-processor is too costly in terms of space and power. However, by design vTPMs must be movable to different platforms to support migration of associated VMs between platforms (\textbf{D2:~\yes}).

In both deployment options, a vTPM has to create its endorsement key at creation time. To establish trust into the EK for a remote verifier, a genuine, primary TPM on the platform (hardware TPM) must attest the trustworthiness of the vTPM's EK (\textbf{S3:~\kinda}).

In case of \cvtpm, the \tpm logic and vTPM instances are executed inside the secure co-processor (\textbf{S6\textbar\cvtpm:~\yes}). Further, the secure co-processor used in \cite{BergerCGPSD06} (an IBM PCIXCC) provides CMOS RAM backed persistent storage. We assume it provides the confidentiality, integrity, and rollback protection of the vTPM states as well as sufficient secure persistent storage to the vTPM clients (\textbf{S1, S2, A1\textbar\cvtpm}:~\yes). The same co-processor also offers facilities for secure counters (\textbf{S4\textbar\cvtpm:}~\tick) and a secure clock (\textbf{S5\textbar\cvtpm:}~\tick).

For \hvtpm the vTPM instances reside in Xen's privileged \emph{dom0}. Thus, their execution is protected from untrusted VMs by only the Xen hypervisor (\textbf{S6\textbar\hvtpm:}~\no), and their state, when stored in persistent storage in \emph{dom0}, is also protected by only the access control and isolation of the hypervisor and \emph{dom0} (\textbf{S1, S2\textbar\hvtpm:}~\kinda). Similar, the protection of any persistent storage offered to vTPM clients depends on the integrity and trustworthiness of \emph{dom0} (\textbf{A1\textbar\hvtpm:}~\no) as does any counter stored in the vTPM state (\textbf{S4\textbar\cvtpm:}~\kinda). A vTPM relies on the platform's clock shared between all vTPMs including untrusted code and not specifically protected (\textbf{S5\textbar\cvtpm:}~\no).
Although vTPM instances are created after the host platform has booted up, a vTPM receives the initial measurement from the underlying hardware \tpm of its platform, which also attests the vTPM trustworthiness, and \emph{dom0} protects the vTPM state from migrating to an untrusted platform (\textbf{D3\textbar\hvtpm:~\kinda}). Further, vTPM instances are created together with their associated VM, hence, allowing the VM to measure its entire bootchain and store the measurements in its vTPM (\textbf{A2:}~\tick).

In case of \cvtpm, the \tpm resides entirely in the IBM PCIXCC, a removable peripheral. Thus, no physical binding to the RTM exists and no authenticity/trustworthiness of the RTM is being ensured (\textbf{D3\textbar\cvtpm:}~\no), hence, an attacker could move the \tpm to an untrusted platform that feeds the \tpm with arbitrary measurements. This situation is very similar to our \ournameplain, which is also removable, and we discuss solutions to this challenge in Section~\ref{bindingToTheDevice}, which might also be applicable to \cvtpm.

The vTPM does not make any assumptions about which stakeholder---user or platform---within the associated VM uses the vTPM and supports, like a regular hardware TPM, multiple hierarchies (\textbf{A3:}~\tick).

\subsection{Intel SGX}
\label{relatedwork:sgx}
Although Intel SGX is not an implementation of a TPM but a solution to allow applications to establish a TEE, \emph{enclave} in SGX jargon (\textbf{S1, S7: \na}), we include it here for comparison because it offers in many dimensions similar protections as a hardware TPM and shares a lot of a TPM's objectives (we mark non-applicable objectives with \na in Table~\ref{tab:countermeasures}). For this work we have only considered stock SGX implementations in Intel processor to keep the comparison on par with other candidates. SGX is currently only supported by desktop and server class Intel processors (\textbf{D1:}~\no) and binds any credentials, like generated and derived keys, and transitively sealed data strictly to the CPU (\textbf{D2:}~\no).

In SGX, \textit{attestation} means verifying that a certain enclave code was initialized correctly and not tampered with by the untrusted host OS. For \textit{remote attestation} in SGX an Intel-provided \textit{Quoting Enclave} provides the facilities to enclaves to do direct anonymous attestation~(DAA) using attestation keys endorsed by Intel (\textbf{S3}:~\tick). The SGX extensions to the CPU measure the enclaves, hence, the enclaves are physically bound to their RTM (\textbf{D3}:~\tick).

SGX supports enclaves in sealing data for storing it on untrusted persistent storage, since enclaves themselves do not have any persistent storage like NV-RAM (\textbf{A1:}~\kinda). In addition, Intel has added support for monotonic counters~\cite{sgxCounter,MateticAKDSGJC17} that allow rollback protection of sealed data (\textbf{S4, S2: \tick}).

It is a processor based technology, so it can fully utilize the clock of the system. But the current SGX implementation does not accommodate a trusted and fine grained clock for the user-level enclaves. There is a  certain API, provided by Intel e.g., \lstinline|get_trusted_time|. But this call can be arbitrarily modified by the untrusted OS, since it requires to make an \textit{OCALL}~\cite{AlderARXIV2018,ChenCCS2017,Liang2018,asokan_keynote,SGXSecureClock}. Moreover, any timing mechanism must account for the fact that the OS can interrupt the enclave at any point in its execution, wait for an arbitrary period of time, and then transparently resume the enclave using \textit{ERESUME} (\textbf{S5: \no}).

Both regular applications and system software can use enclaves and SGX is not restricted to particular stakeholders (\textbf{A3:}~\tick). However, an early firmware initialized enclave is not possible, since the OS is needed for memory management of enclaves (\textbf{A2:}~\cros).

\subsection{Java-card based MTM}
Dietrich and Winter proposed a way of implementing a mobile trusted module (MTM) in a Java-based smart-card for mobile devices~\cite{DietrichW10,EnglandT09}. The implementation is for applications running on mobiles and the TPM communicates through NFC.

The TPM is installed as a set of applets in the Java-card, where a \textit{master-applet} provides services to other applets, like TPM command handling and controls the access to the endorsement key. The actual processing of TPM commands is handled by specific applets implementing those commands.

Although this implementation seems like closest related work to our \stpm, their work described a proof-of-concept prototype and is unfortunately silent about many aspects, such as secure persistent storage, and some functionality is not available, such as attestation of the system or authenticated boot. The Java-card communicates with the system over NFC, so binding the card with the system is not possible and early availability of the trusted module is also not possible before the NFC driver is loaded.

Their implementation provides important insights on the implementation of MTM on mobile devices through a programmable TPM and presented pioneering work, but given the lack of documentation and also differences in engineering (see Section~\ref{sec:systemdesign}), we cannot provide a full and fair comparison with \ournameplain and exclude it from our systematization.

\section{System Design and Security Analysis}
\label{sec:systemdesign}

The main component of \ournameplain is a smart card based implementation of a SIM
TPM. However, to properly work it also requires changes in the bootloader and the operating system
(i.e., Android). In this section, we describe the design and implementation of \ournameplain in more details. We also discuss how our solution solves the shortcomings described in Section~\ref{sec:reqanalysis} and argue about our design's security. Along with the design descriptions, we indicate how the objectives shown in Table 1 are met by \ournameplain.

\subsection{SIM TPM}
\label{sysDesign:simcard}

Modern SIM cards are usually general purpose smart cards running an applet created by the mobile
network provider. The two most prominent smart card technologies are Java Cards and Multos cards.
Both introduce a custom OS (i.e., Java Card OS and Multos OS) and APIs that can be used by
programmers for cryptographic (e.g., encryption, signing) and non-cryptographic (e.g., memory allocation and copy)  operations that are implemented and executed
directly on the microprocessor. Depending on the technology, applets can be programmed in C/C++ (e.g., Multos cards) or in Java (e.g., Java Card).
Additional cryptographic algorithms, not provided by the API, can be implemented in software.

Both card technologies have support for multiple applets. To properly manage them,
cards provide a specialized security manager that is responsible
for installing and deleting of user defined applets. Once an applet is uploaded, the security
manager creates its instance and allows the applet to create necessary objects and allocate memory.
\label{systemDesign:simCardDesign}

\subsubsection{API Limitations of Smart Cards}

As mentioned above, each Smart Card OS provides a card specific API that allows applets to perform extended operations.
 This forces the programmer to use only a predefined set of functions. For example, in case of Java cards the API supports only a subset of the standard Java language and is limited to high level cryptographic operations (e.g., encryption, hashing, signing).
There is no support for mathematical functions like modular multiplication or elliptic curve point addition, which are one of the main building blocks of public key cryptography. In other words, the developer cannot use hardware support for those low-level operations and is limited to software implementations that are inefficient  due to the overhead of the virtualization layer.

Obviously, those limitation do not directly concern TPM commands that only use basic cryptographic operations. Unfortunately, the TPM standard defines a remote attestation scheme that is not supported by the cards
API, because it uses, e.g., zero-knowledge proofs.
This constitutes an interesting engineering problem that we solve.  In particular, we were able to implement  \ournameplain on a Gemalto MultiApp Multos smart card with an Infineon SLE78CLX family microprocessor. This card also helped us achieving process isolation from the general-purpose processor (\textbf{S6: \tick}).
It is worth noting, that in this paper we focused mainly on the Multos API \cite{capi}, because it supports a broader range of functions than the Java card API. In particular, we were able to efficiently implement a remote attestation scheme on-card.



\subsubsection{Smart Cards and TPM Command Parsing}
\label{sTPMandAndroidChanges}

\begin{figure}
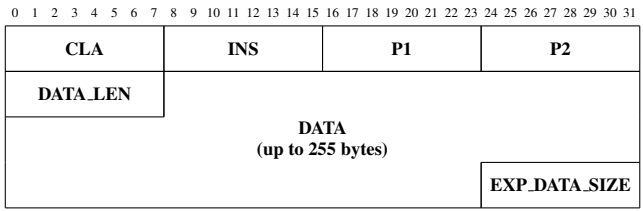

	\begin{bytefield}[bitwidth=0.75em, boxformatting=\centering\scriptsize\bfseries]{32}
		\bitheader{0-31} \\
		\bitbox{8}{CLA} & \bitbox{8}{INS} & \bitbox{8}{P1} & \bitbox{8}{P2}  \\
		\bitbox{8}{DATA\_LEN} & \bitbox[lrt]{24}{} \\
		\wordbox[lr]{1}{DATA \\ (up to 255 bytes)} \\
		\bitbox[lrb]{24}{} & \bitbox{8}{EXP\_DATA\_SIZE}
	\end{bytefield}

	\caption{Generic APDU command structure}
	\label{systemDesign:APDUstructure}
\end{figure}

SIM cards are connected to the main processing unit over a separate bus and available for mobile telephony services (\textbf{D1: \tick}). Smart cards work in a command/response manner, i.e., given an input the card executes the code and returns a response.  The input data is defined by an APDU command (see Figure~\ref{systemDesign:APDUstructure}), which consists of a
class byte (CLA), an instruction byte (INS), two bytes for parameters (P1, P2), one byte for the expected response length, one byte for the data length (DATA\_LEN), and DATA\_LEN bytes of data. The cards' response contains the response data and two bytes that constitute the status word (not shown in the figure). The data field is limited to 255~bytes. There exist an extended length APDU specification that allows for a larger data field but it is not widely
implemented.

In a multi applet system, an APDU command will be forwarded to the currently selected applet. To select an applet, the \verb+SELECT+ APDU command with an unique applet identifier has to be sent to the card. This command is then recognized and executed by the OS.
Once selected, the applet can parse incoming commands according to its work flow. In particular, this means that the developer can use the instruction and parameters bytes to program the behavior of the card. 

The APDU data structure provides a convenient way to communicate with the card. We designed a custom APDU command that implements TPM commands. The data length size of up to 255~bytes is sufficient for the payload sizes of most TPM commands, and for TPM commands with larger payload sizes (e.g., sealed data blobs), we
send the payload split across multiple APDU messages and use the parameter bytes to communicate the card if
more data is to be expected.

\paragraph*{Changes to Android's radio interface:} The Android Radio interface layer (RIL) is responsible for communicating with the device's SIM card. To allow RIL to communicate with \stpm, we introduced a set of TPM commands. We implemented a custom RIL as a shared library, which sends APDU commands as bulk transfer to the \stpm and receives its responses.

\subsubsection{TPM Commands}

We now briefly discuss how we designed the card to handle basic TPM commands related to PCR banks and sealing.
The former case is easy, the applet reserves enough non-volatile memory to store the PCRs. The number of banks is defined by the installation parameter of the TPM applet, which also defines the algorithm we use to extend the PCR (e.g., \verb+SHA1+ or \verb+SHA256+). In a standard setup we use 24 PCRs. For the \verb+TPM_EXTEND+ and \verb+TPM_READ+ commands we used two separate instruction bytes (respectively, \verb+0x10+ and \verb+0x20+) to form the APDU. In both cases the number of the PCRs is given using parameter P1.

To design (un-)sealing on a smart card was a bit harder. Due to the limited input data size, the card has to encrypt/decrypt the input in chunks, which are split across multiple APDU messages. The storage key for sealing is generated by the card after receiving the \verb+TPM_INIT+ command. The key is stored in the non-volatile memory that is allocated during installation of the applet (see next Section~\ref{sec:design:pcrnvram}).

It is worth noting that smart cards can be programmed to execute all TPM commands that require basic cryptographic algorithms, on-card key generation, key agreement, or storing data in volatile/non-volatile memory.
Unfortunately, the privacy-preserving variant of remote attestation (i.e., direct anonymous attestation, DAA)
requires zero-knowledge proofs and other unsupported crypto operations.
What is more, in versions below TPM 2.0 the specification defined only one algorithm for
anonymous attestation~\cite{Brickell:2004:DAA:1030083.1030103}, which is based on groups with hidden order
(i.e., using a RSA modulus) and Camenisch-Lysyanskaya signatures. The TPM 2.0 specification, however, allows for algorithm agility. We leveraged this fact and used a custom
scheme, which we present in the next subsection. 
Here, we only draft the idea behind the scheme, which follows the generic approach used by other DAA schemes: The TPM receives a signature/certificate under its secret
DAA key from an authority. It then uses this secret key to certify its attestation key using a proof. In this
zero-knowledge proof the TPM shows that it knows a certificate under a DAA key and a signature created using
this key under an attestation key. The scheme uses Boneh and Boyen~\cite{eurocrypt-boneh-boyen} signatures and an efficient zero-knowledge proof for the above statement that is made non-interactive using the Fiat-Shamir transformation~\cite{DBLP:conf/crypto/FiatS86}. The main advantage of the scheme is that it can be executed solely by the TPM (i.e., on-card) and does not require any involvement of the host platform.
To further improve efficiency of our scheme, we decided to optimize the workload between commands, i.e., if the \verb+TPM_CREATE+ command recognizes that the TPM is creating an attestation key, it already does some pre-computation for the DAA certification. 

\subsubsection{PCR and NV storage}
\label{sec:design:pcrnvram}
All smart cards implement a small amount of non-volatile storage that can be used for various purposes. This memory of the smart card is by design tamper-resistant and therefore offers memory isolation from the rest of the system (\textbf{S6: \tick}). Modification of this memory is only possible by the applet that reserved it and we reserve some of the NV storage for the \stpm (\textbf{A1: \tick}). Smart cards are equipped with features preventing updates of its internal state by the outside world. To update stored content (e.g., applets), one has
to issue an authorized command to the card manager to update storage or perform applet specific commands, e.g., PCR extension (\textbf{S1: \tick}). Our \ournameplain is equipped with PCR banks that are initialized when power cycling the device and, hence, the SIM card, and can only be changed between power cycles using \verb+PCR_EXTEND+.

System software or user level software can keep a counter containing the current version of the software inside the NV-storage and updates to the counter are only allowed via authorized commands. This provides an easy setup for secure counter and rollback protection (\textbf{S4: \tick}).

\subsubsection{Trustworthy endorsement \& Clock}
Trustworthy endorsement of a TPM is very important. The standard solution is to use an asymmetric encryption key called endorsement key. This key is unique per TPM and should stay alive as long as the TPM is alive. This key differentiates a genuine from a rogue TPM. \ournameplain can achieve secure endorsement by putting a (vendor) certified endorsement key inside its NV-storage and implementing TPM logic that ensures that the private portion of the key is never released to the outside world (\textbf{S3: \tick}).

SIM cards are equipped with a clock pin connected to the baseband processor. Thus, they cannot be clocked higher or lower by an untrusted application or OS. This separate clock helps \ournameplain to work on a different clock frequency not under direct influence of the main processor. What is more, the baseband processor can be used as a secure external clock. In particular, since the baseband processor is by default isolated with a strong security boundary from untrusted code on the platform, it can prepend any APDU command with an APDU command containing the current time (this can also be limited to time-sensitive TPM commands only). This way \ournameplain can be provided with a secure clock (\textbf{S5: \tick}).

\subsubsection{Movability \& Stakeholders}
The other unique feature of the \ournameplain architecture is its movability (\textbf{D2: \tick}). \ournameplain implements the TPM inside the SIM card. So by design, \ournameplain can be transferred to a different device. This creates some interesting use-cases, which we discuss in more details in Section~\ref{switchingSIMcardOrDevice}, but also challenges, which we discuss separately in Section~\ref{bindingToTheDevice}. \ournameplain is not specifically bound to one particular stakeholder and supports the multiple stakeholder model proposed by TCG (\textbf{A3: \tick}), although we think the end-users and their apps are the primary beneficiaries of \ournameplain.

%


\subsection{Our DAA Scheme}

We begin by describing a generic construction/idea of DAA schemes.
First the issuer holding a public key $\gpk$ and a TPM
interact in a join-issue protocol upon which
the TPM obtains a signature (certificate) $cert$ on his secret key $\sk_u$, but 
without revealing the secret key to the issuer. 
Having a certificate on $\sk_u$, the TPM  
 randomizes it
 and produces a signature of knowledge of the secret key
that is certified. 
In some schemes, the TPM uses the secret key $\sk_u$ to 
compute a pseudonym $\nym$ with regards to a basename $\bsn$.
Using the Camenisch-Stadler notation \cite{CSNotation} we can define the 
Signature of Knowledge (SoK) created by the TPM as follows:
\begin{align*}
SoK\{(\sk_u, cert):  \nym = & \hash(\bsn)^{\sk_u} \\
\land & \mathsf{Verify}_{cert, \gpk}(\sk_u) = 1\},
\end{align*}

In the above SoK, $\hash$ denotes a hash function and
$\mathsf{Verify}$ is the verification procedure of the certificate $cert$ on $\sk_u$
with public key $\gpk$.
A common technique used to construct such signatures
of knowledge is to first design a $\Sigma$-protocol, 
which can be used to prove knowledge about linear relations
of exponents. Finally, the designed $\Sigma$-protocol is
transformed into a non-interactive proof (i.e. signature) via the Fiat-Shamir heuristic.
The role of the pseudonym $\nym = D^{\sk_u}$, where $D = \hash(\bsn)$, is to provide 
linkability within a single service identified by the basename $\bsn$.
However, in case no such linkability is necessary, the common
approach is to use a random $D$.

Security of this generic construction follows from the fact that the SoK is extractable and
zero-knowledge. The first property ensures that any security reduction is able to extract a valid
certificate/secret-key pair from a given attestation. Thus, in case the adversary manages to 
create a new TPM the security reduction
can extract the certificate/secret-key pair and break unforgeability of the used signature scheme. 
Note that certificates
are actually signatures of the issuer on the secret keys on TPM DAA keys.
The fact that the TPM can prove in zero-knowledge style that it knows a certificate and that the
pseudonym hides the TPM's identity
is the basis for providing anonymity of the DAA scheme.

We base our construction on the same principles as described above.
The signature scheme that we use is the Boneh-Boyen signature scheme \cite{eurocrypt-boneh-boyen}.
 In scheme~\ref{scheme:daa-join-issue} we 
show how the issuer is able to generate this signature on the TPM's secret key. 
Next in scheme~\ref{scheme:daa-without-GT} we show how to generate a signature of knowledge for the statement
described above and how to verify it. Finally, to prove security we just show that this is in fact a signature of knowledge, i.e. we proof completeness, the existence of an extraction algorithm (to show soundness of the proof) and the existence of a simulator (to show zero-knowledge).

  \begin{scheme}{Direct Anonymous Attestation Setup}{scheme:daa-setup}
\begin{description}
\item{\textsf{Setup}$(\lambda, n)$:}
\begin{enumerate}
\item Choose groups $\GG_1$ and $\GG_2$ of prime order $p$
with a bilinear map $e: \GG_1 \times \GG_2 \rightarrow \mathbb{G}_T$.

\item Choose 
at random  
two group generators $g_1 \rand \GG_1$ and $g_2 \rand \GG_2$, and computes
$g_T \leftarrow e(g_1, g_2)$.

\item Define two hash functions, $\hash$ which 
maps  
into $\ZZ_p$ and $\hash_0$ 
maps  
into $\GG_1$.

\item Select $\sk_\iss \rand \ZZ_p$, set $\gpk_2 \leftarrow  g_2^{\sk_\iss}$
and $\gpk_1 \leftarrow g_1^{\sk_\iss}$.
 
\item Output the public parameters $\crs= (\lambda,\hash, \hash_0, \GG_1, \GG_2, g_1, g_2, g_T, e, \gpk_1, \gpk_2)$.
\end{enumerate} 

\end{description}
\end{scheme} 
 
\begin{scheme}{Direct Anonymous Attestation Join-Issue Protocol (In this protocol the Host acts as a Proxy)}{scheme:daa-join-issue}
\begin{description}
\item{\textsf{Join}$(\crs)$ $\leftrightarrow$ \textsf{Issue}$(\crs, \gsk)$:}
\begin{enumerate}

\item (Platform) The platform generates an additively homomorphic encryption scheme 
$(\skE,\pkE) \exec \keygen(\lambda)$ with message
space that is a superset of $\ZZ_{2^{\lambda+3} \cdot p^2}$.

\item (Platform) The platform chooses $u_i' \rand \ZZ_p$ and computes $U_i' \leftarrow g_1^{u_i}$.

\item (Platform) The platform computes $c_{u_i'} \leftarrow \enc(\pkE, u_i')$.

\item (Platform) The Platform sends $U_i'$, $\pkE$ and $c_{u_i'}$ to the issuer and proves
\begin{align*}
SoK\{\alpha: U_i' = g_1^{\alpha} \land& \; c_{u_i'} = \enc(\pkE, \alpha) \\
\land& \; \alpha \in \ZZ_p \}(\gpk).
\end{align*}

\item (Issuer) The issuer chooses $u_i'' \rand \ZZ_p$, computes
$U_i \leftarrow U_i' \cdot g_1^{u_i''}$.


\item (Issuer) The issuer chooses $b \rand \ZZ_p$,
computes \linebreak
$c_1 \exec  \add(\pkE,\add(\pkE,c_{u_i'}, c_{u_i''}), \enc(\pkE, \sk_\iss))$ and
$c_2\exec \mulC(\pkE, c_1, b)$.

\item (Issuer) The issuer chooses $k \rand \ZZ_{2^{\lambda+2} \cdot p}$,
computes 
$c_{u_i''} \exec \add(c_2,\enc(\pkE, k \cdot p))$ and
$A_i' = g_1^{b}$.

\item (Issuer) The issuer sends $c_{u_i''}$, $A_i'$ 
and $u_i''$ to the platform.

\item (Platform) The platform computes $\sk_{u_i} = u_i' + u_i''$, 
decrypts $t = \dec(\skE, c_{u_i''})$
and computes $A_i  = (A_i')^{t^{-1}} = g_1^{1/(\sk_{u_i} + \sk_\iss)}$.

\item (Platform) The platform verifies that $e(A_i, U_i \cdot \gpk_2) = g_T$ and, if the equation holds, sets $w \leftarrow (A_i, \sk_{u_i})$

%

%


\end{enumerate}
\end{description}
\end{scheme} 
 
  \begin{scheme}{Direct Anonymous Attestation - Signing/Verification}{scheme:daa-without-GT}
  
\begin{description} 

\item{$\signSoK(\crs ,\bsn, w_i, m)$:} 
\begin{enumerate} 

\item (SIM) Parse the secret key as $w_i= (\sk_u, A)$.

\item (SIM)  Compute  $\nym \leftarrow D^{\sk_u}$, where

\begin{itemize}
\item $D \leftarrow \hash_0(\bsn)$, if $\bsn \neq \bot$

\item $D \rand \GG_1$, otherwise.
\end{itemize}

\item (SIM) Choose  
$r \rand \ZZ_p$ 
at random  and compute
$R \leftarrow A^{r}$ 
(so $R= g_1^{r/(\sk_\iss+\sk_u)}$).
Compute \linebreak $B = (g_1 \cdot A^{-\sk_u})^r = g_1^r \cdot R^{-\sk_u}$ (so $B = R^{\sk_\iss}$).

\item (SIM) Compute  
a 
signature of knowledge
\begin{align*}
S \leftarrow SoK\{(\alpha, \beta) :  B =& g_1^\beta \cdot R^{-\alpha} \land \\
D^{\alpha} =& \nym \}(m). 
\end{align*}
Namely, proceed as follows:  
\begin{enumerate}
\item Choose $t_1, t_2 \rand \ZZ_p^2$ 
at random  
and compute
$$
T_1 \leftarrow g_1^{t_2} \cdot R^{-t_1} \text{ ~~and~~ }
T_2 \leftarrow D^{t_1}
$$

\item (SIM) Create a 
challenge:  
$c$ $\leftarrow$ $\hash(\gpk$, $D$, $\nym$, $R$, $B$, $T_1$, $T_2$, $m)$.

\item (SIM) Compute $s_1 \leftarrow t_1 + c \cdot \sk_u$ and $s_2 \leftarrow t_2 + c \cdot r$.

\item (SIM) Set $S \leftarrow (c, s_1, s_2, R, B)$.
\end{enumerate}

\item 
The SIM sends $S$ and $\nym$ (if $\bsn = \bot$ the it sends also $D$) to the Host,
who  outputs the signature $\sigma = (S, \nym)$ (in case $\bsn = \bot$ it adds also $D$ to the signature).
\end{enumerate}

\item{$\verifySoK(\crs, \bsn,  m, \sigma)$:}
\begin{enumerate}

\item Parse the signature $\sigma$ $=$ $(S, \nym)$ ($=$ $(S$, $\nym$, $D)$ in case $\bsn = \bot$).

\item Compute $D \leftarrow \hash_0(\bsn)$ (or in case $\bsn = \bot$ then obtain $D$ from the signature). 

\item Verify the signature of knowledge $S$:

\begin{enumerate}
\item Parse $S$ as $(c$, $s_1$, $s_2, R, B)$.

\item Restore the values 
\begin{align*} 
\tilde{T_1} &\leftarrow   g_1^{s_2} \cdot R^{-s_1}\cdot B^{-c}~.
\\
\tilde{T_2} &\leftarrow  D^{s_1} \cdot \nym^{-c}~.
\end{align*}

\item Check whether  
\begin{equation*}
c \stackrel{?}{=} \hash(\gpk, D, \nym, R, B, \tilde{T_1}, \tilde{T_2}, m)
\end{equation*} 
and
\begin{equation*}
e(R,\gpk) \stackrel{?}{=} e(B,g_2)
\end{equation*}
If yes, then output 1. Otherwise output $0$. 
\end{enumerate} 
\end{enumerate}

\end{description} 
\end{scheme}

\begin{theorem}
Scheme~\ref{scheme:daa-without-GT} is complete.
\end{theorem}

\begin{proof}
Suppose 
that 
the signer 
holds  
a pair $(\sk_u, A) \in \mathbb{Z}_p \times \mathbb{G}_1$
where $A = g_1^{1/(\sk_\iss+\sk_u)}$ and follows the protocol. In this case
\begin{align*}
g_1^{s_2} \cdot R^{-s_1}\cdot B^{-c} &= \\
g_1^{t_2} \cdot  R^{-t_1} \cdot  g_1^{c \cdot r} R^{- c \cdot \sk_u} \cdot B^{-c} &= \\
T_1 \cdot (g_1^{c \cdot r} R^{- c \cdot \sk_u}) \cdot (g_1^{r} \cdot R^{-\sk_u})^{-c} &= \\
T_1 \cdot (g_1^{r} \cdot R^{-\sk_u} )^{c} \cdot (g_1^{r} \cdot R^{-\sk_u})^{-c} &= T_1
\end{align*}
and
\begin{align*}
D^{s_1} \cdot \nym^{-c} &= \\
D^{t_1} \cdot D^{c \cdot \sk_u} \cdot \nym^{-c} &= \\
T_2 \cdot \hash(\bsn)^{c \cdot \sk_u} \cdot \hash(\bsn)^{- c \cdot \sk_u} &= T_2.
\end{align*}
Furthermore, 
\begin{align*}
e(R, \gpk) = e(g_1^{r/(\sk_\iss + \sk_u)}, \gpk) = e(g_1^{r \cdot \gsk/(\sk_\iss + \sk_u)}, g_2)
\end{align*}
and since 
$
\frac{\sk_\iss}{\sk_\iss + \sk_u} = \frac{\sk_\iss + \sk_u - \sk_u}{\sk_\iss + \sk_u} =  1 - \frac{\sk_u}{\sk_\iss + \sk_u}$,
we have
\begin{align*}
e(g_1^{r \cdot \sk_\iss/(\sk_\iss + \sk_u)}, g_2) &= \\
 e(g_1^r \cdot g_1^{-r \cdot \sk_\iss/(\sk_\iss + \sk_u)}, g_2) &= \\
  e(g_1^r \cdot R^{-\sk_u}, g_2) = e(B, g_2).
\end{align*}

\end{proof}

\begin{theorem}
There exists a knowledge extractor for Scheme~\ref{scheme:daa-without-GT}. 
\end{theorem}

\begin{proof} 
By utilizing the rewinding technique, having 
two tuples $(s_1, s_2', c)$ and $(s_1', s_2', c')$
both satisfying the verification equations we
may simply compute  
$\hat{A}$ and $\hat{\sk_u}$ such that
 $\nym = \hash(\bsn)^{\hat{\sk_u}}$
and $e(\hat{A}, \gpk \cdot g_2^{\hat{\sk_u}}) = g_T$.
Denote $\Delta s_i = (s_i - s_i')$ and $\Delta c = (c - c')$.
Since, both $c$ and $c'$ are obtained from the random oracle
with the same input, both tuples need to
give the same $\tilde{T_1}$ and $\tilde{T_2}$.
So we have that
\begin{align*}
\tilde{T_1} &\leftarrow  g_1^{s_2} \cdot R^{-s_1}\cdot B^{-c} = g_1^{s_2'} \cdot R^{-s_1'}\cdot B^{-c}~~\text{and} \\
\tilde{T_2} &\leftarrow  D^{s_1} \cdot \nym^{-c} = D^{s_1'} \cdot \nym^{-c'}.
\end{align*}
and what follows
\begin{align*}
g_1^{\Delta s_2} \cdot R^{- \Delta s_1}  &= B^{\Delta c} ~~\text{and} \\
D^{\Delta s_1}  &= \nym^{\Delta c}.
\end{align*}
So we may compute
$\alpha = \Delta s_1 /\Delta c$ and
$\beta = \Delta s_2/\Delta c$ which
satisfy 
$B = g_1^\beta \cdot R^{-\alpha}$ and $\nym = \hash(\bsn)^{\alpha}$.

Moreover from the verification equation we have that: 
\begin{align*}
&e(R, \gpk) = e(B, g_2) = e(g_1^\beta \cdot R^{-\alpha}, g_2) = \\
&e(g_1,g_2)^{\beta} \cdot e(R,g_2)^{-\alpha}.
\end{align*}
It follows that:
\begin{align*}
e(R, g_2)^{\sk_\iss} \cdot e(R, g_2)^{\alpha} = e(g_1, g_2)^{\beta}.
\end{align*}
and
\begin{align*}
e(R^{\beta^{-1}}, g_2^{\sk_\iss} \cdot g_2^{\alpha \cdot \beta^{-1}}) = e(g_1, g_2).
\end{align*}
Hence we obtain $(\hat{A}, \hat{\sk_u}) = (R^{\beta^{-1}}, \alpha \cdot \beta^{-1})$ 
which satisfies $\nym = \hash(\bsn)^{\hat{\sk_u}}$
and $e(\hat{A}, \gpk \cdot g_2^{\hat{\sk_u}}) = g_T$.
\end{proof}

\begin{theorem}
Assuming $\hash$ is a random oracle
and
given $\pp$, $\bsn$, $\nym \in \GG_1$ and $\gpk', \gpk \in \GG_1 \times \GG_2$ there exists
a simulator which generates values $(R, B, c, s_1, s_2)$ 
 indistinguishable from a signature returned by the original signing procedure 
 of scheme~\ref{scheme:daa-without-GT}.\end{theorem}

\begin{proof}
Let us denote $D \leftarrow \hash(\bsn)$.
First the simulator chooses $\hat{r} \in \ZZ_p$ at random
and computes
$R \leftarrow g_1^{\hat{r}}$ and $B \leftarrow \gpk'^{\hat{r}}$.
At this point it is easy to see that
$e(R, \gpk) = e(B, g_2)$, thus these values meet the verification equation.
Next we will show, that these values are distributed
exactly as in a real signature.
We can write $R = g_1^{\hat{r}} = g_1^{r/(\sk_u + \sk_\iss)}$, for some (unknown) $r \in \ZZ_p$ and (unknown)
$\sk_u \in Z_p$ such that $\nym = D^{\sk_u}$.
Then, the values $B$ looks as follows $B = \gpk'^{\hat{r}} = g_1^{\sk_\iss \cdot \hat{r}} = g_1^{\sk_\iss \cdot r/(\sk_u + \sk_\iss)}$.
Now, it is easy to see that $r \cdot \sk_\iss/(\sk_u + \sk_\iss) = r \cdot (\frac{(\sk_u + \sk_\iss) - \sk_u}{\sk_u + \sk_\iss}) = r \cdot (1 - \sk_u/(\sk_u + \sk_\iss))$.
Therefore, we have $B = g_1^{r \cdot (1 - \sk_u /(\sk_u + \sk_\iss))} = g_1^r \cdot g_1^{ r \cdot \sk_u /(\sk_u + \sk_\iss)} = g_1^r \cdot R^{\sk_u}$, what is exactly the same as in a real signature.

What remains is to compute the t-values.
So, the simulator first 
chooses $c, s_1, s_2 \in \ZZ_p^3$ and
computes
\begin{align*} 
T_1 &\leftarrow   g_1^{s_2} \cdot R^{s_1}\cdot B^{-c}~.
\\
T_2 &\leftarrow  D^{s_1} \cdot \nym^{-c}~.
\end{align*}
Finally, the simulator programs the
random oracle to output $c$ when queries
on $(\gpk$, $D$, $\nym$, $R$, $B$, $T_1$, $T_2$, $m)$.
\end{proof}

\subsection{ATF boot-loader changes}
In Section~\ref{background:atf}, we have briefly introduced ATF and its bootloader chains. In this section we describe the changes we have implemented to enable communication between the bootloader components and the \ournameplain. Figure~\ref{fig:secureBootProcess} can be helpful as a visual aid for understanding.

After turning on the secondary cores on the cold boot path, the processor kicks in the first stage BL1 of the bootloader (\circledb{1}). Current bootloaders are not implemented such as to be able to communicate with a device like a SIM card and to run a command response protocol. Thus, we have extended all the boot-loaders with the capability to communicate with the SIM card via bus communication.
This modification in ATF makes the \stpm already available to the early BL1 stage (\textbf{A2: \tick}).
The bootloader software is capable of translating \tpm commands to APDU commands, sending them to \stpm, receiving responses, and translating them to a meaningful response that can be used to make decisions (e.g., failed/successful PCR extension commands). One thing that needed to be addressed here is that except for BL3-3, all bootloaders are secure mode software (i.e., secure world in TrustZone). So during execution,  \stpm has to be initialized as secure mode hardware to be available to the bootloader. We initialize the \stpm as a secure mode hardware, but after a successful boot chain verification, we switch \stpm to normal mode (of TrustZone). This allows us to maintain normal efficiency in the normal world, since the SIM card functionality (e.g., calls or text messages) is accessed by Android and switching context from normal world to secure world every time before accessing the SIM card in Android can interrupt the normal world execution and would be highly inefficient.

\subsection{Bootstrapping trust for movable \stpm}
\label{bindingToTheDevice}

Parno~\cite{Parno:2008:BTT:1496671.1496680} was first to identify the problem of how to bootstrap trust into a hardware TPM and the possibility of \emph{cuckoo attacks}. A fundamental problem of TPM is that the verifier (e.g., local user) does not know if they are talking to the \emph{intended} (e.g., local) TPM, just that they are talking to \emph{a genuine} TPM. In a cuckoo attack, an attacker that compromised the local platform can exploit this problem and fool the verifier into trusting the compromised platform: the attacker simply relays the verifier's communication to another (remote) TPM on an attacker-controlled platform, which then can attest an arbitrary, trustworthy state to the verifier. The preferred solutions to prevent cuckoo attacks are hardwired channels via a special purpose hardware interface to the on-board TPM or, alternatively, a cryptographically secured verifier-TPM communication where the verifier has knowledge of the public key of the TPM on the intended platform.

However, those solutions make an implicit assumption: Historically TPMs are soldered onto the motherboard, eliminating the issue of ensuring proper binding to the device's root of trust of measurement (RTM), usually in form of an immutable piece of trusted code in the BIOS. Due to this static design a TPM is ensured that the very first received measurement in a chain-of-trust is coming from a trusted, local RTM. Only a sophisticated hardware attack can break this binding. A TPM that is by-design movable, such as our \stpm or the PCI-attached secure co-processor for vTPM~\cite{BergerCGPSD06}, raises an interesting question about how to re-establish this bond between TPM and RTM.

\paragraph*{Lack of chain-of-trust:} Without binding the TPM to a trusted, local RTM, the measurements of any authenticated boot cannot be trusted. An adversary could simply plug the \stpm into an attacker-controlled platform and replay\footnote{The TPM is a passive device to which the measurements have to be provided by its caller.} any desired measurements sequence, i.e., create arbitrary PCR values akin to a TPM reset attack~\cite{TPMHWattack,217652}. This allows the attacker to fool a remote verifier during remote attestation but also to gain access to sealed secrets, whose release is bound to the platform state (i.e., PCR values).

\paragraph*{Binding \stpm and RTM:}

To create a binding between the \stpm and a trusted, local RTM, we need the \stpm to 1)~authenticate the RTM to ensure its a trusted code (e.g., BL1 of ATF); and to 2)~ensure policies (e.g., for data release) and commands (e.g., attestation) are only executed for exactly the platform for which the \stpm stores the measurements. To address those challenges, we identified two possible solutions, using the device's TEE as a proxy to the \stpm 
or using distance bounding protocol. 

\begin{figure}[t]
	\centering
	\includegraphics[width=.90\linewidth]{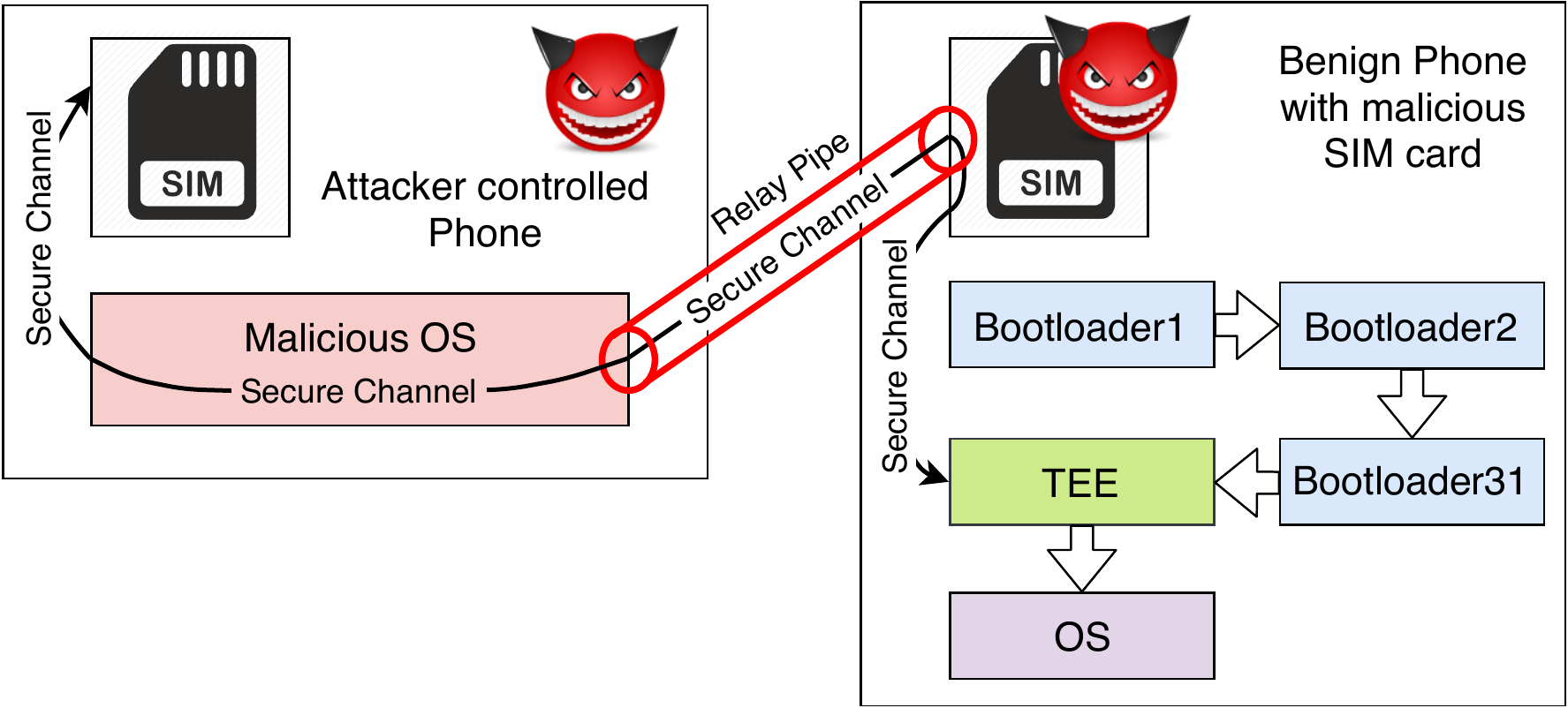}
	\caption{Using TEE as TPM proxy to bind \stpm with RTM and to mitigate the effects of relay attacks.}
	\label{fig:relayAttackTEE}
\end{figure}

\subsubsection*{Using TEE as TPM proxy:}
One way to bind the \stpm with the device's RTM is by leveraging the platform security building blocks of mobile devices and using the TEE as a proxy to \stpm (see Figure~\ref{fig:relayAttackTEE}). On a genuine device with secure boot in place, i.e., BL1 as a trusted RTM, the TEE has exclusive access to device-specific credentials that are certified by the device vendor. Using those credentials, the \stpm and TEE can establish a secure end-to-end channel. In this setup, \stpm will only respond to PCR extensions, attestation requests, or unsealing of encrypted data if the commands come via this secure channel. As a result, an attacker cannot forge arbitrary PCR values without compromising the device-specific key. Further, if the TPM enforces a particular device key, it can ensure that only the intended platform is using the \stpm; however, even without this strict set of device keys, this solution still ensures that any TPM commands, such as releasing data to the host platform, can only come from a genuine mobile platform with an intact secure boot from which it received the measurements.
Considering previously mentioned software-based attacks against TEE (see Section~\ref{relatedWork:fTPM}), an attacker could compromise the TEE to steal the device-specific key and impersonate the TEE to the \stpm. This can be alleviated by using session keys instead of the long-term secret device-specific key for communication between TEE and \stpm, which could be setup during the bootstrapping and, hence, before untrusted code can attack the TEE.
%
%
A drawback of this solution is that the \stpm requires the TEE to be bootstrapped to become itself operational, which prevents an early availability of the \stpm. Since the \stpm is not early available in this setup, ATF's secure boot has to be extended to store the measurements of verified software components and pass those measurements on to the TEE, which then can forward them to the \stpm via the secured channel (\textbf{D3:~\kinda}). 
It should be noted that while this extension to ATF would also provide a solution to the early availability of \ftpm~\cite{RajSWACEFKLMNRS16}, \stpm gives a user-centric solution and additional interesting use-cases in comparison to \ftpm (see Section~\ref{sec:useCases}). We discuss an alternative solution based on distance bounding and no need for a TEE-proxy. 

\subsubsection*{Using Distance Bounding Protocol:}

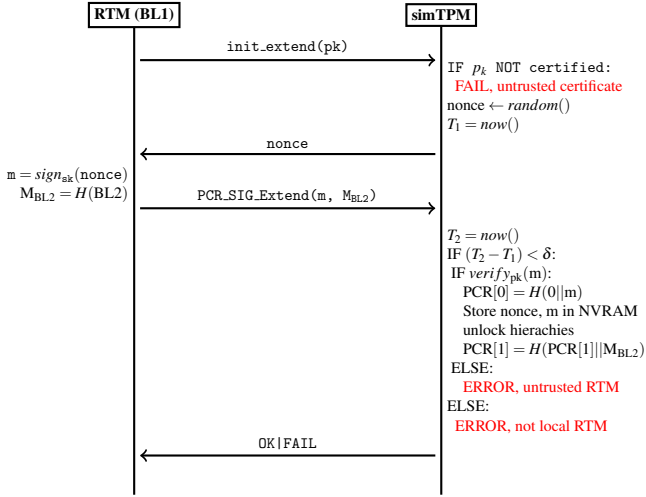
\begin{figure}
	\begin{tikzpicture}[thick,scale=0.6, every node/.style={scale=0.6}]
		\node [draw] (R) {\textbf{RTM (BL1)}};
		\node [draw, right = 3cm of  R] (S) {\textbf{simTPM}};

		\foreach \l in {R, S}
		{\node[below of = \l] (M1\l) {};}
		\draw[->] (M1R) -- (M1S) node[midway,above] {\texttt{init\_extend(pk)}};

		\node[below of = M1R] (L1R) {};
		\node[right, align=left] (L1S) at (L1R -| S){
			\ttfamily
			IF $p_k$ NOT certified:\\
			\indent~~\textcolor{red}{FAIL, untrusted certificate}\\
			$\text{nonce} \leftarrow random()$\\
			$T_1 = now()$\\
		};

		\node(M2S) at (L1S.south west) {};
		\node (M2R) at (M2S -| R) {};
		\draw[<-] (M2R) -- (M2S) node[midway,above] {\texttt{nonce}};

		\node[below of = M2S] (L2S) {};
		\node[left, align=right, yshift=2ex] (L2R) at (L2S -| R){
			\ttfamily
			$\text{m} = sign_{\text{sk}}(\texttt{nonce})$\\
			$\text{M}_{\text{BL2}} = H(\text{BL2})$
		};

		\node(M3R) at (L2R.south east) {};
		\node (M3S) at (M3R -| S) {};
		\draw[->] (M3R) -- (M3S) node[midway,above] {\texttt{PCR\_SIG\_Extend(m, $\text{M}_{\text{BL2}}$)}};

		\node[below of = M3R] (L3R) {};
		\node[right, align=left, yshift=-12ex] (L3S) at (L3R -| S){
			\ttfamily
			$T_2 = now()$\\
			IF $(T_2-T_1)<\delta$:\\
			\indent~IF $verify_{\text{pk}}(\text{m})$:\\
			\indent~~~~$\text{PCR}[0] = H(0||\text{m})$\\
			\indent~~~~Store nonce, m in NVRAM\\
			\indent~~~~unlock hierachies\\
			\indent~~~~$\text{PCR}[1] = H(\text{PCR}[1]||\text{M}_{\text{BL2}})$\\
			\indent~ELSE:\\
			\indent~~~~\textcolor{red}{ERROR, untrusted RTM}\\
			ELSE:\\
			\indent~~\textcolor{red}{ERROR, not local RTM}\\
		};

		\node(M4S) at (L3S.south west) {};
		\node (M4R) at (M4S -| R) {};
		\draw[<-] (M4R) -- (M4S) node[midway,above] {\texttt{OK|FAIL}};

		\foreach \l in {R, S}
		{\node[below of = M4\l] (bot\l) {};}
		\foreach \l in {R, S}
		{\draw (\l) -- (bot\l);}

	\end{tikzpicture}

	\caption{Prototypical distance bounding protocol for binding local RTM (BL1) and \stpm}
	\label{fig:distanceBounding}
\end{figure}

Another way to bind the \stpm with its RTM is by using a distance bounding (DB) protocol~\cite{brandsdistance,bengiosecure,beth1990identification}. Distance bounding is widely used for card-based payment systems. When a credit card is punched to the card reader, the reader runs a distance bounding protocol to check the proximity of the card to prevent a possible relay attack. We are facing the opposite scenario, in which the card is trying to assert the proximity of the device where the communication partner, here the RTM, resides.

\paragraph*{Prototypical distance bounding:} We assume, the device vendors equipped the \stpm with certificates for their device-specific keys, which allows a verifier to distinguish trusted code with access to such secrets (e.g., early bootstages, like BL1, or the TEE) from untrusted code, like the host OS or regular apps. To assert the proximity of the RTM, only the very first measurement provided to the \stpm, i.e., the measurement by the BL1 (RTM) of BL2, has to be checked for proximity. After that, the chain of trust of an authenticated boot will transitively extend this trust into the locality of the RTM. Figure~\ref{fig:distanceBounding} illustrates a prototypical protocol for our scenario. We consider a two-step PCR extension by the RTM for verifying the proximity. First, the RTM provides the public key \emph{pk} of its device-specific key (or a key derived from it) to the TPM, which then can verify the authenticity of the RTM using the vendor-supplied certificate. Afterwards, as in other distance bounding protocols, the \stpm (verifier) challenges the RTM (prover) with a nonce to which the RTM replies with the signed nonce value (using the authenticated private key) as well as the PCR extension arguments. If this reply of the signed nonce is received within a time threshold $T$ and the signature verifies, \stpm assumes the RTM to be local and extends the PCR with the supplied measurement value $M_{BL2}$; if either condition fails, the \stpm aborts. For robustness of the protocol, the challenge-response can be repeated $N$ times to decrease the chances of a legitimate, local RTM failing the threshold.

\paragraph*{Prototypical setup:} In general, calculating the threshold for distance bounding is difficult, because various factors can influence the response time. For instance, jitters of the network over which the verifier and prover communicate, interrupts of the prover's computation, cache and memory delays, etc. might introduce a high uncertainty of the expectable response time. At first glance, our particular scenario seems very favorable for a distance bounding protocol, since the prover (RTM) is the BL1 and hence has exclusively control of the CPU without interrupts or interference of an OS; and the RTM is connected to the SIM card over a 480~mbps USB~2.0 bus, in modern devices even via a USB~3.0 bus with 5~gbps, with no parallel transfers, providing favorable circumstances for a challenge-response protocol and small error-margin in which an attacker has to fall for a successful, undetected relay attack~\cite{drimer2007keep,cremers2012distance,mauw2018distance}.

The SIM card is connected to the phone through a reader, which is directly connected to the baseband processor. The reader powers the smart card and provides it with the baseband's clock. The clock duty cycle shall be between 40\% and 60\% of the period during stable operation~\cite{ETSITS151011}. Modern smart cards support clock stop to allow preservation of power, which an attack could use to tamper with the verifier's perception of time. However, this feature can be disabled by initializing the card as clock stop not allowed by setting the \textit{VERIFY CHV} command to 0. Disabling this feature will increase the phone's battery consumption, but not in a significant amount, since the maximum current consumption of an idle SIM card should not exceed $200 \mu A$.

The SIM card and the reader connection are in a contact connection and generally interfaces within $20ns$~\cite{BCT4303,LTC4558}. The reader connects to the baseband processor through Non-Level-Shifted bidirectional I/O (see Figure~\ref{fig:siminterface}). The connection in our test setup goes through an USB 2.0 bus with 480~mbps. The baseband processor connects to the CPU with a system bus (varies between a USB 2.0 bus and USB 3.0 bus) with a speed of at least 480~mbps. Communication between SIM card and the CPU via this bus interface lies within a range of $35ns$ to $72ns$.

\begin{figure}[t]
	\centering
	\includegraphics[width=1.0\linewidth]{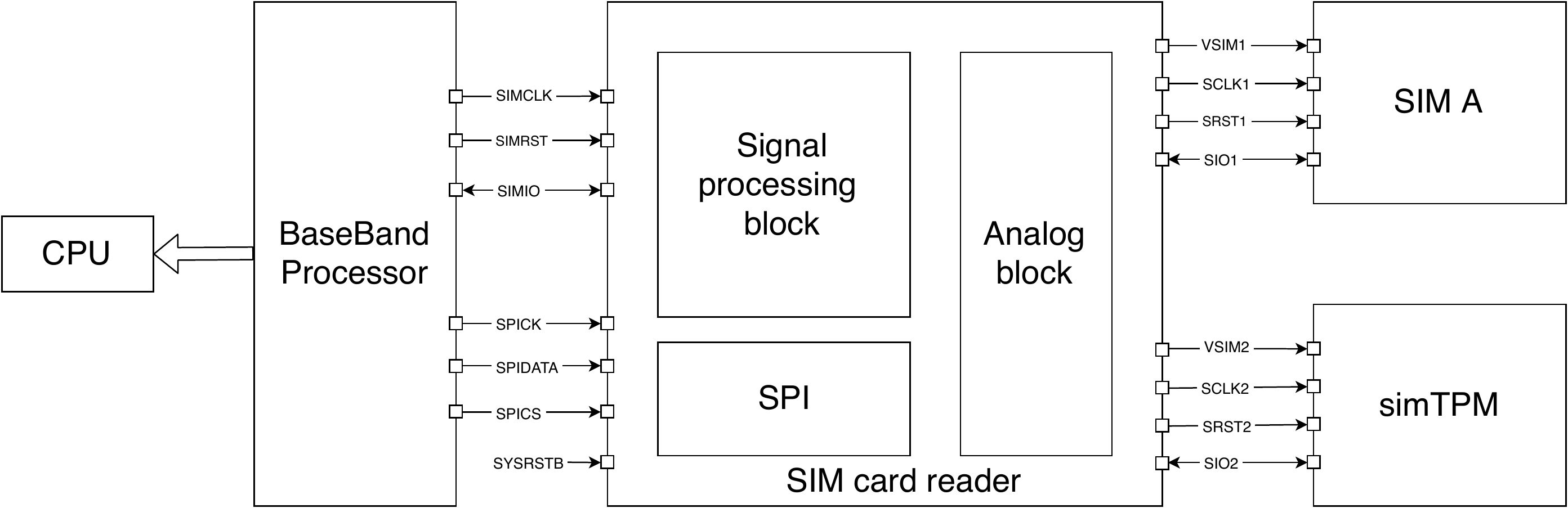}
	\caption{SIM card interfacing with the phone. Communication with the SIM card goes through the SIM card reader to the baseband processor to the CPU}
	\label{fig:siminterface}
\end{figure}

\begin{figure}[t]
	\centering
	\includegraphics[width=1.0\linewidth]{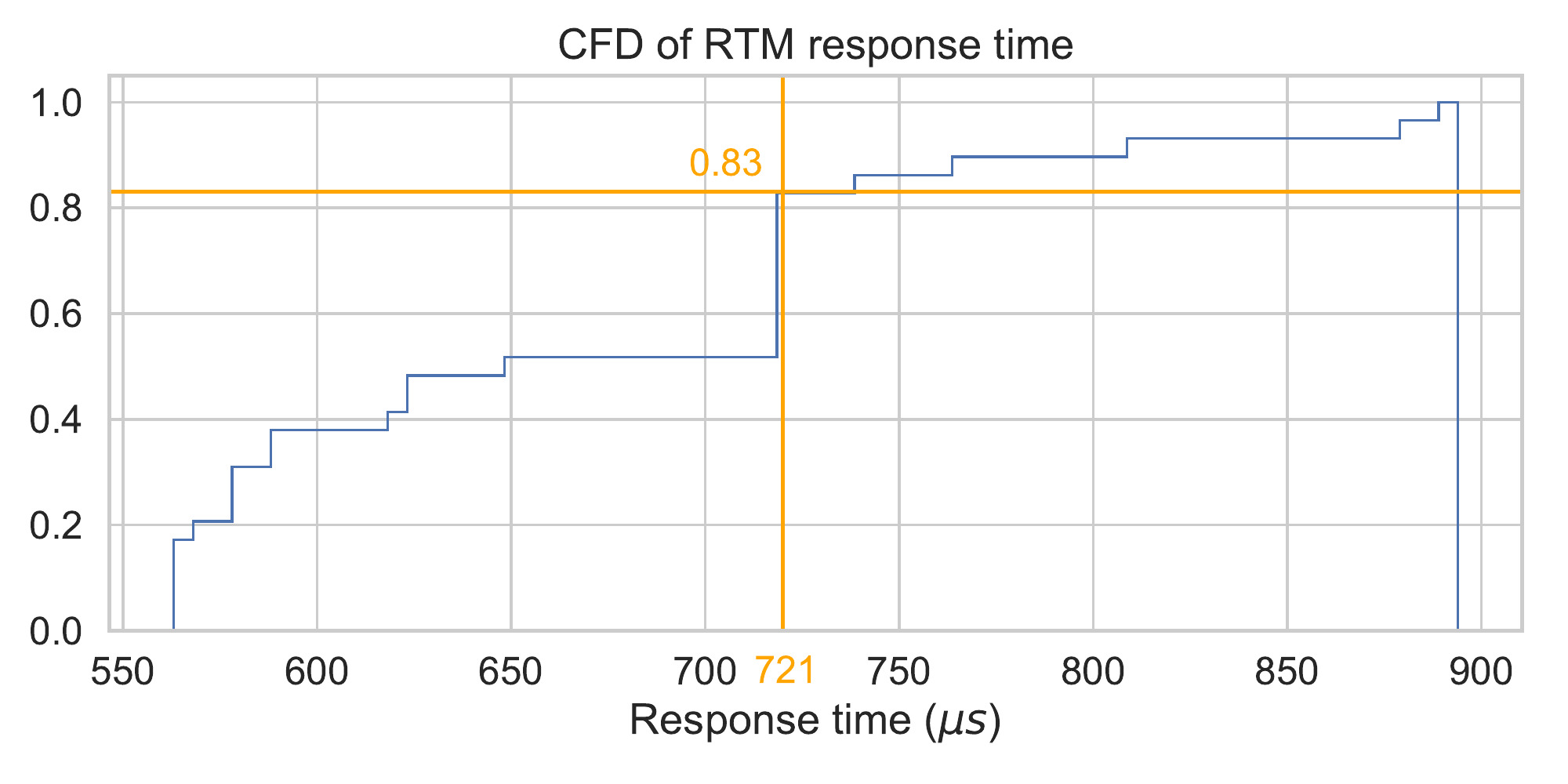}
	\caption{Cumulative frequency distribution of the RTM response time in our measurements ($N=30$)}
	\label{fig:cfd_rtm}
\end{figure}

\paragraph*{Measurements and setting threshold:} We conducted a series of measurements on our test device to evaluate the feasibility of distance bounding to bind the RTM and \stpm. We measured 30 times\footnote{A single measurement requires $\approx$5min, since only a single measurement per power-cycle is possible on our test device.} the speed of the prover (RTM) for calculating the response to the challenge (64~bits nonce) using ECC with the NIST P-256 curve. In our test, the responses took between $563 \mu s$ and $894\mu s$, and the average response time was $669.759 \pm 49.804 \mu s$ for a confidence level of 99\%. Figure~\ref{fig:cfd_rtm} shows the CFD of the RTM response time, where 83\% of all responses were $\leq721\mu s$ and 93\% of all responses were $\leq812\mu s$. From our dataset the success chance of the distance bounding protocol $P_{DB}$ for a single round is the cumulative probability sampled over the frequency distribution in Figure~\ref{fig:cfd_rtm}. If were to set the threshold $T$ for successful distance bounding to 721 $\mu s$:

\begin{center}
	$P_{DB}$ = Pr[x $\leq$ 721] = $\sum_{i=563}^{721}$ Pr[$x=i$] $\approx$ 0.83
\end{center}
where 563 $\mu s$ is the lowest latency in our dataset. Going below 721 $\mu s$ reduces the probability of a successful bounding protocol for legitimate devices, i.e., to 0.52 for a threshold of 649$\mu s$. To increase the chances for local RTM to pass the distance bounding check, a successful verification usually requires that the response is below $T$ for a sufficient fraction $f$ of the responses, means at-least $f \times n$ out of $n$ responses should arrive within $T$. When modeling the challenge-response game as binomial distribution and requiring $f \times n$ responses within 721 $\mu s$ out of $n$ responses (i.e., success probability $P_{DB} = 0.83$), the cumulative probability distribution is:

\begin{center}
	Pr[$x \geq fn$] = $\sum_{i=fn}^{n}$ $\left({\begin{array}{c} n \\ i \end{array}}\right) (p)^i (1-p)^{n-i}$ where $p \in \{P_{DB}\}$
\end{center}

Figure~\ref{fig:bernoulli_cdf} shows the success probabilities for different choices of $f$ and $n$. An optimal choice minimizes $n$ (lower overall runtime overhead for the protocol) while maximizing Pr[$x \geq fn$] and minimizing the chance of the attacker to successfully relay. We have observed from our dataset that setting $f$ = 0.47 for $n$ = 30 (i.e., 14 out of 30 runs) offers a success rate Pr[$x \geq fn$]=0.99999724049 for local RTM.

\begin{figure}[t]
	\centering
	\includegraphics[width=1.0\linewidth]{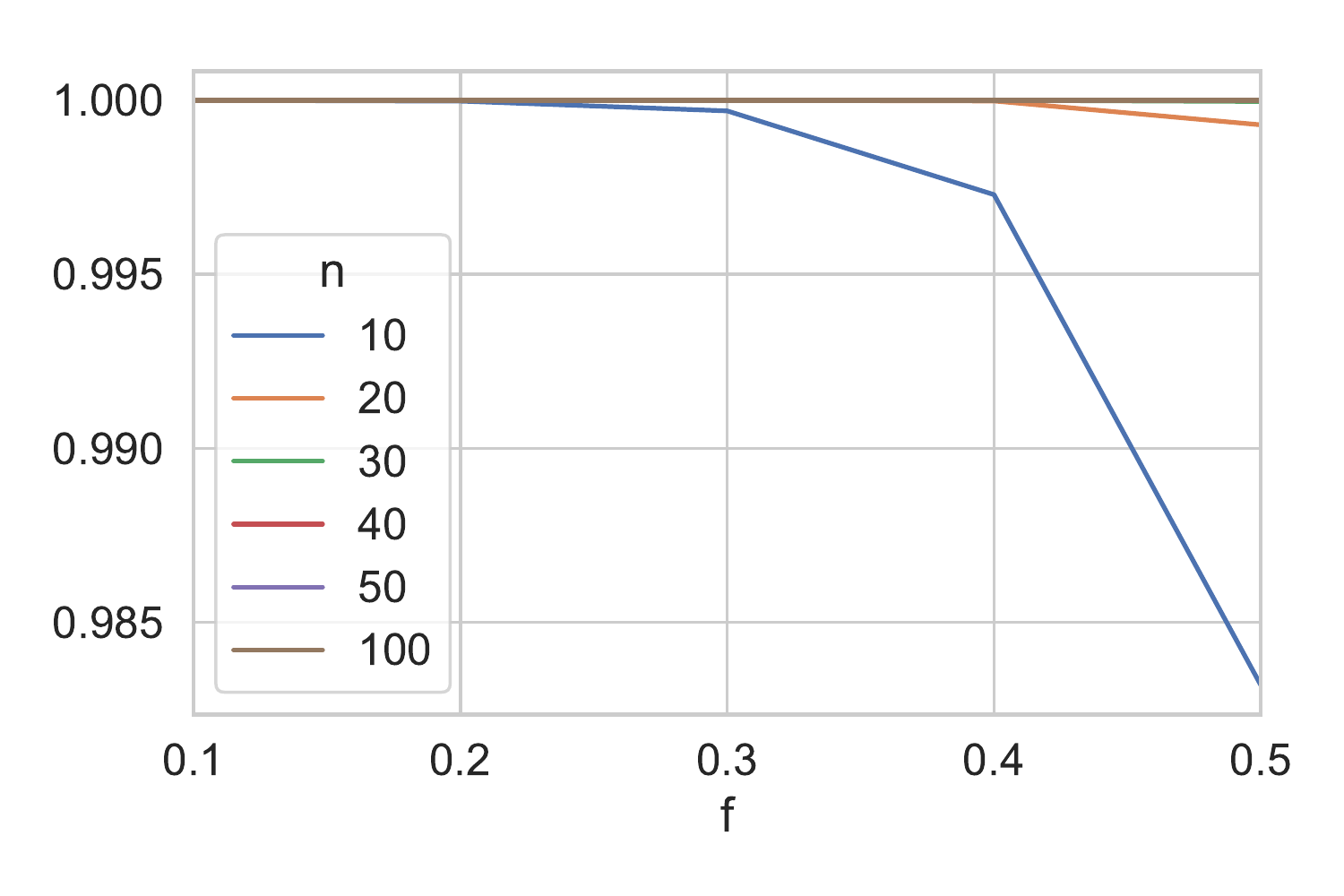}
	\caption{Success probability of local RTM for distance bounding depending on $f$ and $n$ for $T=721\mu s$ ($p=0.83$)}
	\label{fig:bernoulli_cdf}
\end{figure}

\paragraph*{Attacker chances:} The APDU package for the challenge is 112~bits and for the response 304~bits, which are transferred virtually instantly between verifier and prover~($\leq1\mu s$). Thus, the response time measured in Figure~\ref{fig:cfd_rtm} consists virtually only of the processing time of the RTM, which an attacker cannot speed up (see Figure~\ref{fig:relayAttackTEE}). As a consequence, if an attacker requires more than $721-563=158 \mu s$ to relay the challenge and the response, the relay attack has no chance of winning, since the RTM in our tests required at least $563 \mu s$ to compute the response. Assuming a packet size of 55~bytes (minimal Ethernet frame size, IP header, and UDP package with 1~byte payload for the nonce/response), the attacker needs at least a relay bandwith of $\approx 5.87$ mbps to have \emph{any} chance of winning, which is a very reasonable assumption. Hence, attacks against this distance bounding are feasible. From our measurements it is hard to concretely model the attacker, however, the attack chance is already 0.1\% when relaying via Ethernet and an IP network (55~bytes datasize) with a bandwith of $\approx 49$ mbps, or when relaying only the APDU data of 14~bytes (e.g., via a custom build connection) with $\approx 10$ mbps.


\subsection{Security analysis}

Lastly, we analyze the security of \ournameplain in comparison to the closest solutions \ftpm and hardware TPM, specifically considering the deployment of our TPM on a SIM card.

\paragraph*{Off-chip protection:} As mentioned in Section~\ref{sec:reqanalysis}, \ftpm depends on the integrity of the secure world, which has been under attack recently~\cite{bootstomp,MachiryGSSSWBCK17,TZDowngrade,ExploitingTZTEE,bitsPlease,BHTZ,BHRose,nailgun:sp19}. Our \stpm implements an off-board \tpm on the SIM card and, like a discrete TPM, is physically isolated from untrusted code. This provides a stronger protection of the \stpm's trusted computing base, however, we cannot fully exclude potential software attacks against the SIM card software. For instance, in the past smart cards have exhibited bugs~\cite{MostowskiP08} like hidden commands, buffer overflows, weaknesses of cryptographic protocols~\cite{KarstenBH13}, or malicious applets~\cite{KarstenBH13}. Further, like a hardware TPM, \stpm is connected via a bus, which makes it prone to advanced bus attacks~\cite{conf/uss/Kauer07, TPMHWattack,tpmgenie} that, however, are considered outside the attacker model for consumer grade hardware like the TPM.

\paragraph*{SIM card cloning:}

Deployment on a SIM card also raises the concern of card cloning~\cite{gsmcloning}, which could easily enable impersonation attacks or theft of credentials. However, driven by the interests of telecommunication companies, modern SIM cards come with anti-cloning defenses that mitigate this attack vector~\cite{capi}.

\paragraph*{SIM fraud attack:}
The sim swap fraud centers around exploiting a mobile phone operator's ability to seamlessly port a telephone number to a new SIM. This scam begins with a fraudster gathering details about the victim, either by phishing emails, by social engineering or by buying victim's details through organized criminals. Obtaining the details of the victim, the fraudster uses social engineering techniques to convince the telephony company to port the victims number to a fraudulent sim card owned by the fraudster and start receiving all sms's including secure communication (e.g., banking OTP).
In our design the TPM is not dependent on the SIM telephony functionalities. simTPM works as a local co-processor with desirable attributes. An attacker can port the telephony services to a fraudulent SIM card, but not the TPM state, as it is bound to the local SIM-card and would require explicit migration policies to other (SIM)TPM.

\paragraph*{Side-channel attacks:}

To be compliant with the \tpm 2.0 specification, the hardware has to implement cryptographic functions
that are resilient to timing-based side-channel attacks. There exists a similar requirement for smart cards, which are designed to be resistant against various types of side-channel attacks. Thus, \ournameplain immediately benefits from the security features of the underlying smart card.

However, a motivated attacker can easily move \ournameplain to a controlled environment and mount different active side-channel attacks, such as clock frequency, heat measurement, probing~\cite{Kasper:2009:ESA:1698100.1698110}, fault injection~\cite{KommerlingK99}, or power analysis~\cite{Liu:2015:STH:2993033.2993061,MessergesD99,MessergesDS99}. While similar attacks have been shown against ARM TrustZone (e.g., \cite{197209,203864}) and discrete TPM chips~\cite{tpm_extractkey}, deploying the TPM on a removable card might ease mounting those attacks. Nevertheless, it should be noted that such sophisticated hardware attacks are not only strenuous, exorbitant, and inconsistent, but also beyond the protection that a consumer grade security chip can offer.


\section{Performance Evaluation}
\label{sec:evaluation}

\begin{figure*}
	\centering
	\includegraphics[width=\textwidth]{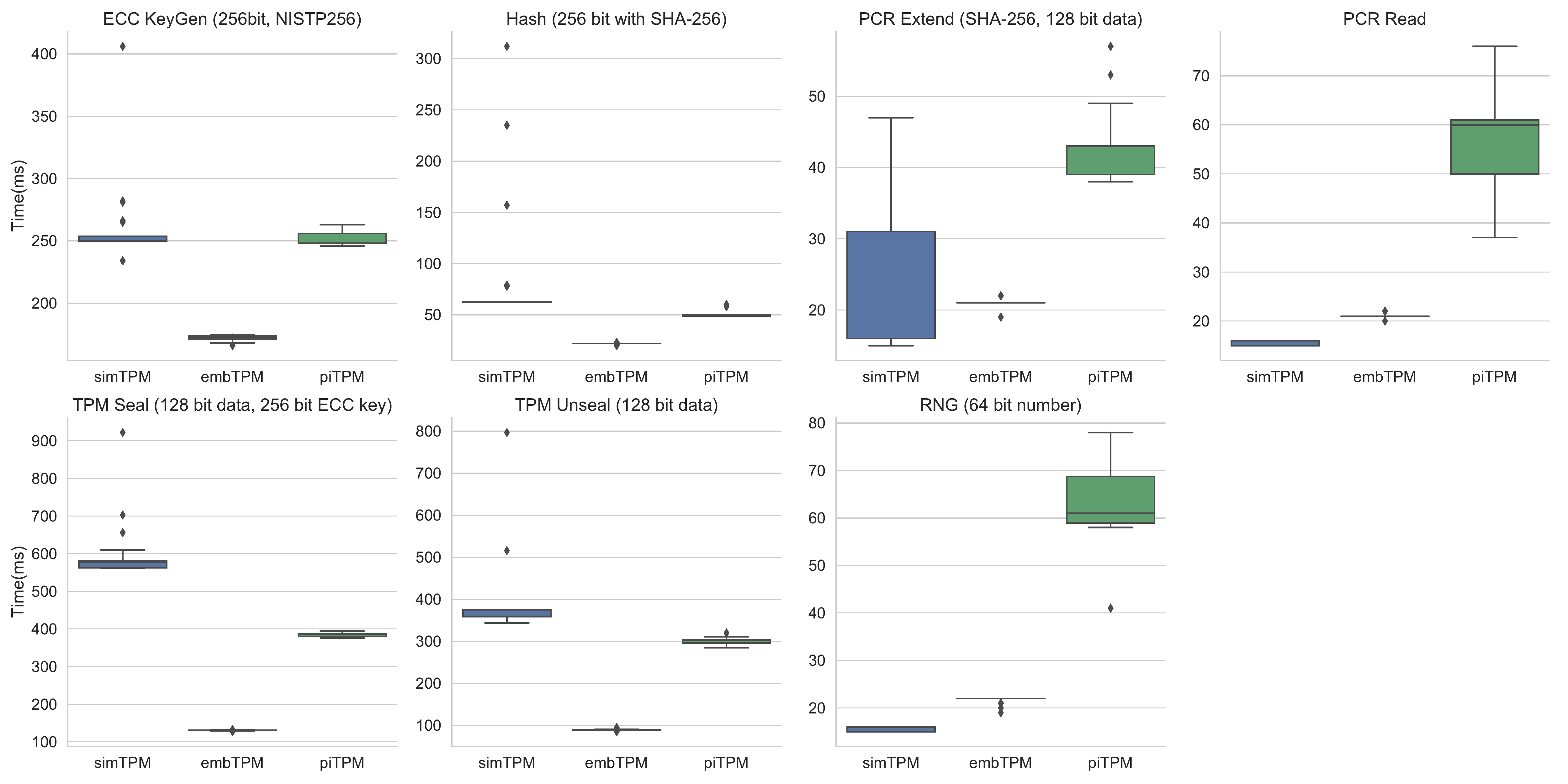}
	\caption{Performance comparison (in $ms$) of different \tpm commands for \stpm and an Infineon SLB 9670 TPM2.0 on a Raspberry-Pi and Lenovo laptop.}
	\label{fig:measurements}
\end{figure*}

We evaluate the performance of \ournameplain on a HiKey960 board in comparison with a hardware \tpm. We focus on the most frequent commands executed by a \tpm, i.e., key generation, sealing/unsealing of data, extending/reading a PCR, generating random bytes, and computing a hash value of an input. Beside \ournameplain we prepared two test setups equipped with an Infineon SLB 9670 TPM chip. One of these two test benches is a plug-able TPM on a Raspberry-Pi (\textit{piTPM}) and the other one is an embedded TPM on a standard Lenovo laptop (\textit{embTPM}). 


The \textit{\atpm} consists of an Infineon TPM SLB 9670 Iridium add-on board for Raspberry Pi with a 900 MHz quad-core cortex-A7 CPU as primary CPU and 1GB RAM. 
In the \textit{\etpm}, the \tpm is embedded on a board equipped with an Intel core i7-7820 2.90GHz CPU and 32 GB RAM. 
The TPM SLB 9670 chip is capable of asymmetric cryptography---ECC, ECC BN-256, ECC NIST P-256, ECC256, ECDH, RSA1024, RSA2048---and symmetric cryptography, like HMAC, SHA-1, SHA-256.

For \ournameplain, the platform runs Android P on a Hikey960 board. The board consists of a Kirin 960 Soc, 
an 8 core ARM big.LITTLE CPU and 3GB RAM. The Android and boot-loader implementations used in those setups and the SIM-card implementation are explained in Section~\ref{sTPMandAndroidChanges}. The hardware setup is detailed in Table~\ref{tab:hardwareSetup}.

We have used a TSS implementation by IBM~\cite{TPM2TSS} to communicate with the Infineon TPM. The results of our benchmarks are summarized in Figure~\ref{fig:measurements}. All results come from 50 measurements per command per device. We report the 95\% confidence intervals.

\subsection{Test cases and results}

\paragraph*{Key generation:}

We measured the time to generate a
256-bit ECC key and output the public part of the key. Our implementation of \stpm creates the key on average in $257\pm 8.03ms$, comparable to the \atpm performance ($253\pm 1.25ms$), but slower than the \etpm ($172\pm 0.61ms$).

\paragraph*{Create hash:}

We measured the time it takes for the TPM to hash 256~bits of input data with SHA-256 and output the digest. \atpm ($50\pm 0.76ms$) and \etpm ($21\pm 0.16ms$) outperform the \stpm ($72\pm 10.13ms$) by a factor of 1.44 and 3.42, respectively.

\paragraph*{Extending and reading a PCR:}

We evaluated the PCR extend and read commands.
The former allows to extend the PCR with a new value, while the latter command is used to read the current value of a PCR. We use SHA-256 as hash algorithm and a 128 bit string as input value. For PCR extension, \stpm ($24\pm 2.66ms$) is on par with \etpm ($21\pm 0.11ms$), however, exhibits a higher instability of the performance. For reading PCRs, \stpm ($15\pm 0.15ms$) is the fastest implementation, followed by \etpm ($21\pm 0.13ms$). \atpm is the slowest implementation in both cases ($41\pm 1.22ms$ and $57\pm 2.58ms$) and exhibits an unstable performance, too.

\paragraph*{Sealing and unsealing data:}

The \tpm seal command takes a byte array, attaches a policy, encrypts it with a TPM storage key, and returns a blob to the caller. When unsealing, the \tpm takes an encrypted blob, checks the policy, and decrypts the blob if the policy is satisfied by the \tpm state. For our performance measurement we used 128~bits input data, a 256-bit ECC sealing key with ECIES, and an empty policy. The \etpm is the fatest solution for sealing and unsealing ($130\pm 0.27ms$ and $89\pm 0.46ms$) and outperforms our \stpm ($588\pm 18.55ms$ and $376\pm 22.30ms$) by a factor of 4.52 and 4.22, respectively.

\paragraph*{Random number generation:} We use the TPM to generate a 64~bit random number. Our \stpm is the fastest solution ($15\pm 0.14ms$), followed by \etpm ($21\pm 0.17ms$) and then \atpm ($63\pm 1.63ms$).

\subsection{Discussion of performance}

Our test results show that there is no clear winner among our test systems. \stpm as well as \etpm excel for some commands and we would argue that our \stpm prototype shows a competitive performance. Unfortunately, the implementation of Infineon SLB 9670 TPM is not publicly available, thus commenting on the exact reasons for those differences would result in speculations. If we would venture to speculate, potential reasons for the differences could be the different communication buses. \etpm has a dedicated bus communication with the onboard processor and a faster processor, while \atpm is running on a Raspberry Pi and is connected over GPIO with lower bandwidth. On the other hand, \stpm is connected through the USB bus. Moreover, our \stpm implementation uses only the publicly available APIs of the smart card OS, which provide only an indirect access to hardware level commands. Hence, a vendor-supported implementation with direct access to the microprocessor would improve in efficiency.

The \ftpm is unfortunately not available, precluding a direct comparison in our test suite; however, our observations for the \etpm speed are comparable to those reported by Raj et al.~\cite{RajSWACEFKLMNRS16}, although it is unclear which hardware TPM they evaluated. An \ftpm, unsurprisingly, outperforms any other tested implementation here---e.g., slowest \ftpm in \cite{RajSWACEFKLMNRS16} was between 2.4--15.12 times faster than the fastest hardware TPM---since it is executed on the ARM Cortex main application processor, whereas discrete TPMs use slower microprocessors, as does our \stpm.

\begin{table*}
	\caption{Hardware setup for performance evaluation}
	\label{tab:hardwareSetup}
	\begin{center}
		\begin{tabular}{| c | c | c | c | c | c | c | c | c |}
			\hline
			\textbf{Setup Alias}     &    \textbf{Baseboard}        &    \textbf{Clock-speed}        &    \textbf{Processor}    &    \textbf{RAM}	&	\textbf{OS}	&	\textbf{TSS}	\\ \hline \hline
			\atpm    &    RaspberryPI 2B    &    900 MHz    &   ARM Cortex A7    &    1 GB	&	Raspbian OS		&	IBM TSS2.0	   \\ \hline
			\etpm   &	  Lenovo 20J6CTO1WW		 & 	  2.90 GHz   &    Intel Core i7-7820  &  32 GB	& 	Ubuntu		&	IBM TSS2.0  \\ \hline
			\stpm &   Hikey960		 &	  903 MHz	 & 	  Kirin960		&  	3 GB	&	Android P	&	Custom   \\ \hline
		\end{tabular}
	\end{center}
\end{table*}

\section{Use Cases}
\label{sec:useCases}

We discuss briefly how \stpm fits into the trusted computing landscape and explain scenarios that are of particular interest when \stpm and \ftpm co-exist. 

\subsection{Multiple stakeholder model}
\label{sec:usecase:msm}

The TPM specifications~\cite{TPM2Spec} as well as the obsolete Mobile Trusted Module~(MTM) specifications~\cite{TCGMTMSpe} acknowledged the fact that a trusted platform might have multiple stakeholders. In particular, mobile platforms are not considered under the full management of the user, but critical mobile network management is the domain of the mobile carrier/network operator and the device vendor has high interest in keeping highest privileged operations (e.g., TEE and OS) under their control. The old and new TCG specifications define recommended capabilities and various implementation alternatives to allow multiple stakeholders to safely coexist. For instance, the MTM specification clearly differentiates between remote stakeholders and local stakeholders, each with their own TPM under their control. This concept is reflected in the recommended capabilities for a mobile TPM2.0~\cite{TCGMTMSpec2}, which advise the isolation between stakeholders and their resources and policy-based authorization of stakeholder sensitive data. To realize this multiple stakeholder model, the reference architecture outlines different implementation alternatives. For instance, multiple TPMs within a protected environment like TEE, or virtual TPMs supported by a hypervisor~\cite{BergerCGPSD06}, where stakeholders are isolated from each other based on the compartmentalization provided by the TEE's trusted OS or the hypervisor, respectively.

Our particular setting also fits well into the defined multiple stakeholder model: two distinct TPMs co-exist, each with a distinct affinity to a different stakeholder. The \ftpm is by design designated to the platform stakeholder (i.e., device manufacturer) and it is bound to the device through the device-specific credentials within the TEE (e.g., eFuses) from which \ftpm derives its endorsement key and to which it anchors its key/storage hierarchies. For instance, the \ftpm described in \cite{RajSWACEFKLMNRS16} is designated entirely to the platform and its services. In contrast, the \stpm is designated to the end-user. This intuition is based on the observation that users use the SIM to authenticate themselves to the mobile network and rather stick to one SIM (i.e., phone number) while changing more frequently the device. Moreover, users have to explicitly authenticate themselves to the SIM card, i.e., their mobile carrier issued PIN.
In this setting we are going beyond the initial proposals by the TCG reference architecture by actually assigning two distinct stakeholders to two \emph{physically} separated TPM instances, SIM card versus TEE. 

\subsection{Switching SIM card or device}
\label{switchingSIMcardOrDevice}
Since the SIM card is removable and exchangeable, two scenarios have to be considered: the user switches devices but keeps the SIM card, or the user keeps the device and switches to a new SIM card. How this affects migration of the user data protected with the \stpm is summarized in Table~\ref{table:useCaseDatMigration} and explained in the following.

\begin{table}[t]
	\caption{Migrating user data when switching SIM card or device}
	\label{table:useCaseDatMigration}
	\footnotesize
	\begin{center}
		\begin{tabularx}{\linewidth}{| c | >{\centering\arraybackslash}X | >{\centering\arraybackslash}X |}
			\hline
			 & Data bound to device & Data \textbf{not} bound to device \\
			\hline
			New SIM card & Key duplication & Key duplication \\ \hline
			New device   & \texttt{TPM\_Authorize} key policy  & --- \\ \hline
		\end{tabularx}
	\end{center}
\end{table}

\paragraph*{Switching device:}
When switching the mobile device and migrating the user data to a new device, the complexity of the operation is dependent on whether the user bound any data to the device. For instance, during secure boot, BL1 has access to device-specific information like the \boardid that uniquely identifies the current platform (or potentially values derived from the device-specific vendor key). This \boardid (like derived values) can be included in the measurements collected during secure boot (see Section~\ref{bindingToTheDevice}) and allow the \stpm to bind data or keys to this particular platform.\footnote{Assuming a bond between the RTM and \stpm was established.} If the user did not bind any data/keys to the platform, no further action is required beyond moving the SIM card to the new phone. The entire \stpm state including the key hierarchy is inherently migrated to the new device and can be used to decrypt the user data---i.e., a form of \emph{portable sealed storage}. If the data is bound to the \boardid, a new feature of TPM2.0 called \lstinline{TPM_Authorize} has to be used to avoid the problem of \emph{"brittle policies."} Without \lstinline{TPM_Authorize}, the user data would be bound to one particular \boardid and could never be decrypted on another device. With \lstinline{TPM_Authorize} different possible \boardid values can be signed off as valid for a successful verification of the platform state and, hence, decryption of data migrated with the SIM card. The valid \boardid values can be signed off by the user to endorse a new phone to which data should be migrated, or by another entity, like the mobile carrier or the user's employer in BYOD settings.

\paragraph*{Switching SIM card:}
\label{sec:usecase:switchingSIM}
If the user switches the SIM card and hence moves to another \stpm, all user data has to be migrated to the new SIM card, i.e., the necessary \stpm keys have to be moved to the new \stpm. Independent of whether the user data is bound to the device or not, switching the SIM card requires the \stpm keys used for securing the data to be duplicated to the new \stpm. This is an example scenario for TPM2.0 key duplication to migrate keys and associated data to another TPM and is supported by \ournameplain.

\medskip
\noindent
The bottom line of those two scenarios is that a user that wants to keep the option to migrate data secured with the \stpm to both new SIM cards and new devices should use duplicable keys with \lstinline{TPM_Authorize}.


\section{Discussion}
\label{sec:discuss}

The \ftpm~\cite{RajSWACEFKLMNRS16} is the incumbent deployment for a TPM on mobile devices and was part of the Windows Phone platform. However, it was designed primarily for vendor services and did not specifically target the end-user. In this work, we add to the landscape of mobile trusted computing and advocate using the dormant hardware capabilities of SIM cards to provide (additional) TPM support on mobile devices. Our systematization of related works shows that a \stpm can take a niche among the existing works and, in particular, inherently avoids problems of TEE-based deployments (e.g., protected state or secure clock) that currently require compromises and modifications to the TPM specification (e.g., "dark period" or cooperative checkpointing of \ftpm) or that make additional hardware requirements (e.g., replay-protected memory blocks). On the other hand, a movable TPM raises the challenge of how to bind the TPM and the platform RTM. In this work, we proposed using the unique features of mobile devices---secure boot and TEE with device-specific, certified keys---to address this challenge. However, we find that this problem also affects prior solutions, like a vTPM based on a PCI-attached secure co-processor, and our solution might give insights into how to establish the TPM-RTM binding in those prior works.

 Our \stpm implementation is based on a physical SIM card, thus it is currently not suitable for phones using eSIM (e.g., Apple iPhone). However, eSIM solutions are supported by separate hardware modules (such as JEDEC SON-8) and it might be worthwhile to investigate how those modules can be extended to implement a full TCG compliant TPM2.0.

Recently, Google introduced their Titan chip~\cite{googleTitan} as part of their Nexus~3 phones, which shows the need for hardware-backed security features in addition to TEE-based implementations on mobile end-user devices. Similar to the \stpm, Titan chip also provides hardware-backed security for system operations like verified booting as well as a hardware-implemented keystore for apps and users. But Titan is exclusive for Google devices, whereas our \stpm is portable between mobile devices and provides TPM2.0 compliant features. Since implementation details are yet unknown, we excluded the Titan chip from our systematization in Section~\ref{sec:reqanalysis}.

\section{Conclusion}
\label{sec:conclusion}

In this paper we proposed \ournameplain, a hardware-based TPM implementation for mobile devices using the SIM card. Performance evaluation of our prototype shows that our implementation is comparable with an existing discrete TPM chip. Thus, we think \stpm is a practical solution to add user-centric trusted computing technology to mobile devices without the need to add hardware. A particular challenge of a movable TPM is the binding between TPM and the device RTM, which we addressed through a TEE-proxy or a distance bounding protocol.
Future work includes a more detailed and formal write-up of the custom DAA scheme we used in our prototype, since it is particularly fitting for implementation on a smart card. Also future implementations of \ournameplain in industrial IoT or automotive settings for hardware based attestation could be worthwhile to pursue.

\section{Acknowledgment}
\label{sec:ack}

We are grateful to N. Asokan in \say{particular} for his insightful suggestions. We are also thankful to the anonymous reviewers for their valuable reviews.

This work is supported by the German Federal Ministry of Education and Research(BMBF) through funding for the Center for IT-Security, Privacy and Accountability (CISPA)(AutSec/FKZ: 16KIS0753) and the CISPA-Stanford Center for Cybersecurity (FKZ: 16KIS0762).

{\normalsize \bibliographystyle{acm}
\bibliography{usenix2019_v2}

\begin{thebibliography}{10}

\bibitem{hikey960}
Hikey960 android development board.
\newblock \url{https://www.96boards.org/product/hikey960/}.
\newblock Accessed: 02.08.2018.

\bibitem{TPM2TSS}
Ibm's tpm 2.0 tss.
\newblock \url{https://sourceforge.net/projects/ibmtpm20tss/}.
\newblock Accessed: 06.08.2018.

\bibitem{bitsPlease}
Trustzone downgrade attack opens android devices to old vulnerabilities.
\newblock
  \url{http://bits-please.blogspot.com/2015/03/getting-arbitrary-code-execution-in.html},
  Mar. 2015.
\newblock Accessed: 02.08.2018.

\bibitem{AlderARXIV2018}
{\sc Alder, F., Asokan, N., Kurnikov, A., Paverd, A., and Steiner, M.}
\newblock S-faas: Trustworthy and accountable function-as-a-service using intel
  {SGX}.
\newblock {\em CoRR abs/1810.06080\/} (2018).

\bibitem{TBBR}
{\sc {Arm Limited}}.
\newblock Trusted board boot design guide.
\newblock
  \url{https://github.com/ARM-software/arm-trusted-firmware/blob/master/docs/trusted-board-boot.rst},
  Mar. 2018.
\newblock Accessed: 04.08.2018.

\bibitem{asokan_keynote}
{\sc Asokan, N.}
\newblock On secure resource accounting for outsourced computation, 2018.
\newblock Invited keynote at 3rd Workshop on System Software for Trusted
  Execution (SysTEX 2018).

\bibitem{bengiosecure}
{\sc Bengio, S., Brassard, G., Desmedt, Y.~G., Goutier, C., and Quisquater,
  J.-J.}
\newblock Secure implementation of identification systems.
\newblock {\em Journal of Cryptology 4}, 3 (1991), 175--183.

\bibitem{ExploitingTZTEE}
{\sc Beniamini, G.}
\newblock Getting arbitrary code execution in trustzone's kernel from any
  context.
\newblock
  \url{https://googleprojectzero.blogspot.com/2017/07/trust-issues-exploiting-trustzone-tees.html},
  July 2017.
\newblock Accessed: 02.08.2018.

\bibitem{BergerCGPSD06}
{\sc Berger, S., Cáceres, R., Goldman, K.~A., Perez, R., Sailer, R., and van
  Doorn, L.}
\newblock vtpm: Virtualizing the trusted platform module.
\newblock In {\em Proc.\ 15th USENIX Security Symposium (SEC '06)\/} (2006),
  USENIX Association.

\bibitem{beth1990identification}
{\sc Beth, T., and Desmedt, Y.}
\newblock Identification tokens—or: Solving the chess grandmaster problem.
\newblock In {\em Conference on the Theory and Application of Cryptography\/}
  (1990), Springer, pp.~169--176.

\bibitem{eurocrypt-boneh-boyen}
{\sc Boneh, D., and Boyen, X.}
\newblock Short signatures without random oracles.
\newblock In {\em Advances in Cryptology - EUROCRYPT 2004: International
  Conference on the Theory and Applications of Cryptographic Techniques\/}
  (2004), Springer.

\bibitem{tpmgenie}
{\sc Boone, J.}
\newblock Tpm genie: Attacking the hardware root of trust for less than \$50,
  2018.
\newblock Accessed: 02/13/2019.

\bibitem{brandsdistance}
{\sc Brands, S., and Chaum, D.}
\newblock Distance-bounding protocols.
\newblock In {\em Workshop on the Theory and Application of Cryptographic
  Techniques on Advances in Cryptology (EUROCRYPT '93)\/} (1994), Springer.

\bibitem{Brickell:2004:DAA:1030083.1030103}
{\sc Brickell, E., Camenisch, J., and Chen, L.}
\newblock Direct anonymous attestation.
\newblock In {\em Proc.\ 11th ACM Conference on Computer and Communication
  Security (CCS '04)\/} (2004), ACM.

\bibitem{BCT4303}
{\sc {Broadchip}}.
\newblock {BCT4303 Dual Sim card controller}.
\newblock
  \url{www.chinesechip.com/files/2015-03/912ed043-de27-4e8a-95f7-c009ad22dd92.pdf}.
\newblock Last accessed: 22/01/19.

\bibitem{CSNotation}
{\sc Camenisch, J., and Stadler, M.}
\newblock Efficient group signature schemes for large groups (extended
  abstract).
\newblock In {\em Proc.\ 17th Annual International Cryptology Conference on
  Advances in Cryptology (CRYPTO '97)\/} (1997), Springer.

\bibitem{SoK}
{\sc Chase, M., and Lysyanskaya, A.}
\newblock On signatures of knowledge.
\newblock In {\em Proc.\ 26th Annual International Conference on Advances in
  Cryptology (CRYPTO'06)\/} (2006), Springer.

\bibitem{ChenCCS2017}
{\sc Chen, S., Zhang, X., Reiter, M.~K., and Zhang, Y.}
\newblock Detecting privileged side-channel attacks in shielded execution with
  d{\'e}j\`{a} vu.
\newblock In {\em Proceedings of the 2017 ACM on Asia Conference on Computer
  and Communications Security\/} (New York, NY, USA, 2017), ASIA CCS '17, ACM,
  pp.~7--18.

\bibitem{TZDowngrade}
{\sc Cimpanu, C.}
\newblock Trust issues: Exploiting trustzone tees.
\newblock
  \url{https://www.bleepingcomputer.com/news/security/trustzone-downgrade-attack-opens-android-devices-to-old-vulnerabilities/},
  Sept. 2017.
\newblock Accessed: 02.08.2018.

\bibitem{journals/iacr/CostanD16}
{\sc Costan, V., and Devadas, S.}
\newblock Intel sgx explained.
\newblock {\em IACR Cryptology ePrint Archive 2016}, 086 (2016), 1--118.

\bibitem{cremers2012distance}
{\sc Cremers, C., Rasmussen, K.~B., Schmidt, B., and Capkun, S.}
\newblock Distance hijacking attacks on distance bounding protocols.
\newblock In {\em Proc.\ 33rd IEEE Symposium on Security and Privacy (SP
  '12)\/} (2012), IEEE Computer Society.

\bibitem{DietrichW10}
{\sc Dietrich, K., and Winter, J.}
\newblock Towards customizable, application specific mobile trusted modules.
\newblock In {\em Proc.\ 5th ACM workshop on Scalable trusted computing (STC
  '10)\/} (2010), ACM.

\bibitem{drimer2007keep}
{\sc Drimer, S., Murdoch, S.~J., et~al.}
\newblock Keep your enemies close: Distance bounding against smartcard relay
  attacks.
\newblock In {\em Proc.\ 16th USENIX Security Symposium (SEC '07)\/} (2007),
  USENIX Association.

\bibitem{EkbergKA14}
{\sc Ekberg, J., Kostiainen, K., and Asokan, N.}
\newblock The untapped potential of trusted execution environments on mobile
  devices.
\newblock {\em IEEE Security Privacy 12}, 4 (2014), 29--37.

\bibitem{Ekberg13}
{\sc Ekberg, J.-E.}
\newblock {\em Securing Software Architectures for Trusted Processor
  Environments}.
\newblock PhD thesis, Aalto University, Helsinki, Finland, 2013.

\bibitem{bugiel09:stc}
{\sc Ekberg, J.-E., and Bugiel, S.}
\newblock Trust in a small package: Minimized {MRTM} software implementation
  for mobile secure environments.
\newblock In {\em Proc.\ 4th ACM workshop on Scalable trusted computing (STC
  '09)\/} (2009), ACM.

\bibitem{EnglandT09}
{\sc England, P., and Tariq, T.}
\newblock Towards a programmable tpm, 2009.

\bibitem{ETSITS151011}
{\sc {ETSI}}.
\newblock {TS 151 011 V4.15.0 (2005-06) Technical Specification Digital
  cellular telecommunications system (Phase 2+); Specification of the
  Subscriber Identity Module - Mobile Equipment (SIM-ME) interface (3GPP TS
  51.011 version 4.15.0 Release 4)}.
\newblock
  \url{https://www.etsi.org/deliver/etsi_ts/151000_151099/151011/04.15.00_60/ts_151011v041500p.pdf}.
\newblock Last accessed: 22/01/19.

\bibitem{DBLP:conf/crypto/FiatS86}
{\sc Fiat, A., and Shamir, A.}
\newblock How to prove yourself: Practical solutions to identification and
  signature problems.
\newblock In {\em Proc.\ on Advances in cryptology (CRYPTO '86)\/} (1987),
  Springer.

\bibitem{217652}
{\sc Han, S., Shin, W., Park, J.-H., and Kim, H.}
\newblock A bad dream: Subverting trusted platform module while you are
  sleeping.
\newblock In {\em Proc.\ 27th USENIX Security Symposium (SEC' 18)\/} (2018),
  USENIX Association.

\bibitem{sgxCounter}
{\sc {Intel}}.
\newblock Software guard extensions sdk: sgx\_create\_monotonic\_counter.
\newblock
  \url{https://software.intel.com/en-us/sgx-sdk-dev-reference-sgx-create-monotonic-counter},
  May 2018.

\bibitem{SGXSecureClock}
{\sc {Intel Developer Zone}}.
\newblock {Platform Service Enclave and ME for Intel Xeon Server}.
\newblock
  \url{https://software.intel.com/en-us/forums/intel-software-guard-extensions-intel-sgx/topic/806502}.
\newblock Last accessed: 20/05/19.

\bibitem{JangKKKK15}
{\sc Jang, J.~S., Kong, S., Kim, M., Kim, D., and Kang, B.~B.}
\newblock {SeCReT}: Secure channel between rich execution environment and
  trusted execution environment.
\newblock In {\em Proc.\ 22nd Annual Network and Distributed System Security
  Symposium (NDSS '15)\/} (2015), The Internet Society.

\bibitem{Kasper:2009:ESA:1698100.1698110}
{\sc Kasper, T., Oswald, D., and Paar, C.}
\newblock Information security applications.
\newblock Springer, 2009, ch.~EM Side-Channel Attacks on Commercial Contactless
  Smartcards Using Low-Cost Equipment, pp.~79--93.

\bibitem{conf/uss/Kauer07}
{\sc Kauer, B.}
\newblock Oslo: Improving the security of trusted computing.
\newblock In {\em Proc.\ 16th USENIX Security Symposium (SEC '07)\/} (2007),
  USENIX Association.

\bibitem{KommerlingK99}
{\sc K{\"{o}}mmerling, O., and Kuhn, M.~G.}
\newblock Design principles for tamper-resistant smartcard processors.
\newblock In {\em Proceedings of the 1st Workshop on Smartcard Technology,
  Smartcard 1999, Chicago, Illinois, USA, May 10-11, 1999\/} (1999).

\bibitem{TPMHWattack}
{\sc Lawson, N.}
\newblock Tpm hardware attacks.
\newblock \url{https://rdist.root.org/2007/07/16/tpm-hardware-attacks/}, July
  2007.
\newblock Accessed: 06.08.2018.

\bibitem{Liang2018}
{\sc Liang, H., and Li, M.}
\newblock Bring the missing jigsaw back: Trustedclock for sgx enclaves.
\newblock In {\em Proceedings of the 11th European Workshop on Systems
  Security\/} (New York, NY, USA, 2018), EuroSec'18, ACM, pp.~8:1--8:6.

\bibitem{LTC4558}
{\sc {Linear Technology}}.
\newblock {LTC4558 - Dual SIM/Smart Card Power Supply and Interface}.
\newblock
  \url{https://www.analog.com/media/en/technical-documentation/data-sheets/4558fa.pdf}.
\newblock Last accessed: 22/01/19.

\bibitem{197209}
{\sc Lipp, M., Gruss, D., Spreitzer, R., Maurice, C., and Mangard, S.}
\newblock Armageddon: Cache attacks on mobile devices.
\newblock In {\em Proc.\ 25th USENIX Security Symposium (SEC' 16)\/} (2016),
  USENIX Association.

\bibitem{Liu:2015:STH:2993033.2993061}
{\sc Liu, J., Yu, Y., Standaert, F.-X., Guo, Z., Gu, D., Sun, W., Ge, Y., and
  Xie, X.}
\newblock Small tweaks do not help: Differential power analysis of milenage
  implementations in 3g/4g usim cards.
\newblock In {\em Proc.\ 20th European Symposium on Research in Computer
  Security (ESORICS 2015)\/} (2015), Springer.

\bibitem{MachiryGSSSWBCK17}
{\sc Machiry, A., Gustafson, E., Spensky, C., Salls, C., Stephens, N., Wang,
  R., Bianchi, A., Choe, Y.~R., Kruegel, C., and Vigna, G.}
\newblock Boomerang: Exploiting the semantic gap in trusted execution
  environments.
\newblock In {\em Proc.\ 24rd Annual Network and Distributed System Security
  Symposium (NDSS '17)\/} (2017), The Internet Society.

\bibitem{capi}
{\sc {MAOSCO Limited}}.
\newblock Multos standard c-api.
\newblock \url{https://www.multos.com/uploads/CAPI.pdf}, 2016.
\newblock Accessed: 02.08.2018.

\bibitem{MateticAKDSGJC17}
{\sc Matetic, S., Ahmed, M., Kostiainen, K., Dhar, A., Sommer, D., Gervais, A.,
  Juels, A., and Capkun, S.}
\newblock {ROTE}: Rollback protection for trusted execution.
\newblock In {\em Proc.\ 26th USENIX Security Symposium (SEC' 17)\/} (2017),
  USENIX Association.

\bibitem{mauw2018distance}
{\sc Mauw, S., Smith, Z., Toro-Pozo, J., and Trujillo-Rasua, R.}
\newblock Distance-bounding protocols: Verification without time and location.
\newblock In {\em Proc.\ 39th IEEE Symposium on Security and Privacy (SP
  '18)\/} (2018), IEEE Computer Society.

\bibitem{McCunePPRI08}
{\sc McCune, J.~M., Parno, B.~J., Perrig, A., Reiter, M.~K., and Isozaki, H.}
\newblock Flicker: An execution infrastructure for tcb minimization.
\newblock In {\em Proc.\ 3rd ACM SIGOPS/EuroSys European Conference on Computer
  Systems (Eurosys '08)\/} (2008), ACM.

\bibitem{McGillionDNA15}
{\sc McGillion, B., Dettenborn, T., Nyman, T., and Asokan, N.}
\newblock Open-tee -- an open virtual trusted execution environment.
\newblock In {\em Proc.\ IEEE Trustcom/BigDataSE/ISPA - Volume 01 (TRUSTCOM
  '15)\/} (2015), IEEE Computer Society.

\bibitem{MessergesD99}
{\sc Messerges, T.~S., and Dabbish, E.~A.}
\newblock Investigations of power analysis attacks on smartcards.
\newblock In {\em Proceedings of the 1st Workshop on Smartcard Technology,
  Smartcard 1999, Chicago, Illinois, USA, May 10-11, 1999\/} (1999).

\bibitem{MessergesDS99}
{\sc Messerges, T.~S., Dabbish, E.~A., and Sloan, R.~H.}
\newblock Power analysis attacks of modular exponentiation in smartcards.
\newblock In {\em Cryptographic Hardware and Embedded Systems, First
  International Workshop, CHES'99, Worcester, MA, USA, August 12-13, 1999,
  Proceedings\/} (1999).

\bibitem{windows_tpm}
{\sc {Microsoft}}.
\newblock Secure the windows 10 boot process.
\newblock
  \url{https://docs.microsoft.com/en-us/windows/security/information-protection/secure-the-windows-10-boot-process},
  Oct. 2017.
\newblock Last accessed: 08/06/18.

\bibitem{MostowskiP08}
{\sc Mostowski, W., and Poll, E.}
\newblock Malicious code on java card smartcards: Attacks and countermeasures.
\newblock In {\em Smart Card Research and Advanced Applications, 8th {IFIP}
  {WG} 8.8/11.2 International Conference, {CARDIS} 2008, London, UK, September
  8-11, 2008. Proceedings\/} (2008).

\bibitem{nailgun:sp19}
{\sc Ning, Z., and Zhang, F.}
\newblock Understanding the security of arm debugging features.
\newblock In {\em Proc.\ 40th IEEE Symposium on Security and Privacy (SP
  '19)\/} (2019), IEEE Computer Society.

\bibitem{KarstenBH13}
{\sc Nohl, K.}
\newblock Rooting sim cards, blackhat usa 2013.
\newblock
  \url{https://media.blackhat.com/us-13/us-13-Nohl-Rooting-SIM-cards-Slides.pdf},
  2013.

\bibitem{OkamotoUchiyama}
{\sc Okamoto, T., and Uchiyama, S.}
\newblock A new public-key cryptosystem as secure as factoring.
\newblock In {\em Advances in Cryptology (EUROCRYPT '98)\/} (1998), Springer.

\bibitem{Paillier}
{\sc Paillier, P.}
\newblock Public-key cryptosystems based on composite degree residuosity
  classes.
\newblock In {\em Proc.\ 17th International Conference on Theory and
  Application of Cryptographic Techniques (EUROCRYPT '99)\/} (1999), Springer.

\bibitem{Parno:2008:BTT:1496671.1496680}
{\sc Parno, B.}
\newblock Bootstrapping trust in a "trusted" platform.
\newblock In {\em Proc.\ 3rd Conference on Hot Topics in Security
  (HOTSEC'08)\/} (2008), USENIX Association.

\bibitem{RajSWACEFKLMNRS16}
{\sc Raj, H., Saroiu, S., Wolman, A., Aigner, R., Cox, J., England, P., Fenner,
  C., Kinshumann, K., Löser, J., Mattoon, D., Nyström, M., Robinson, D.,
  Spiger, R., Thom, S., and Wooten, D.}
\newblock ftpm: A software-only implementation of a tpm chip.
\newblock In {\em Proc.\ 25th USENIX Security Symposium (SEC' 16)\/} (2016),
  USENIX Association.

\bibitem{bootstomp}
{\sc Redini, N., Machiry, A., Das, D., Fratantonio, Y., Bianchi, A., Gustafson,
  E., Shoshitaishvili, Y., Kruegel, C., and Vigna, G.}
\newblock Bootstomp: on the security of bootloaders in mobile devices.
\newblock In {\em Proc.\ 26th USENIX Security Symposium (SEC' 17)\/} (2017),
  USENIX Association.

\bibitem{BHRose}
{\sc Rosenberg, D.}
\newblock Reflections on trusting trustzone.
\newblock
  \url{https://www.blackhat.com/docs/us-14/materials/us-14-Rosenberg-Reflections-on-Trusting-TrustZone.pdf},
  2014.
\newblock Accessed: 02.08.2018.

\bibitem{BHTZ}
{\sc Shen, D.}
\newblock Exploiting trustzone on android.
\newblock
  \url{https://www.blackhat.com/docs/us-15/materials/us-15-Shen-Attacking-Your-Trusted-Core-Exploiting-Trustzone-On-Android-wp.pdf},
  2015.
\newblock Accessed: 02.08.2018.

\bibitem{203864}
{\sc Tang, A., Sethumadhavan, S., and Stolfo, S.}
\newblock {CLKSCREW}: Exposing the perils of security-oblivious energy
  management.
\newblock In {\em Proc.\ 26th USENIX Security Symposium (SEC' 17)\/} (2017),
  USENIX Association.

\bibitem{tpm_extractkey}
{\sc Tarnovsky, C.}
\newblock Deconstructing a 'secure' processor, 2010.

\bibitem{tpm_chromium}
{\sc {The Chromium Projects}}.
\newblock {TPM Usage}.
\newblock \url{https://www.chromium.org/developers/design-documents/tpm-usage}.
\newblock Last accessed: 08/06/18.

\bibitem{TCGMTMSpe}
{\sc {Trusted Computing Group}}.
\newblock Mobile phone work group mobile trusted module specification.
\newblock
  \url{https://trustedcomputinggroup.org/resource/mobile-phone-work-group-mobile-trusted-module-specification/},
  2010.

\bibitem{TPM1Spec}
{\sc {Trusted Computing Group}}.
\newblock Tpm main part 1 design principles.
\newblock
  \url{https://trustedcomputinggroup.org/wp-content/uploads/TPM-Main-Part-1-Design-Principles\_v1.2\_rev116\_01032011.pdf},
  2011.

\bibitem{TCGMTMSpec2}
{\sc {Trusted Computing Group}}.
\newblock Tpm 2.0 mobile reference architecture specification.
\newblock
  \url{https://trustedcomputinggroup.org/resource/tpm-2-0-mobile-reference-architecture-specification/},
  2014.

\bibitem{TPM2Spec}
{\sc {Trusted Computing Group}}.
\newblock Tpm 2.0 library specification.
\newblock
  \url{https://trustedcomputinggroup.org/resource/tpm-library-specification/},
  2016.

\bibitem{gsmcloning}
{\sc Wagner, D.}
\newblock Gsm cloning.
\newblock \url{http://www.isaac.cs.berkeley.edu/isaac/gsm.html}.
\newblock {Last accessed: 02/13/19}.

\bibitem{Winter08}
{\sc Winter, J.}
\newblock Trusted computing building blocks for embedded linux-based arm
  trustzone platforms.
\newblock In {\em Proc.\ 3rd ACM workshop on Scalable trusted computing (STC
  '08)\/} (2008), ACM.

\bibitem{googleTitan}
{\sc Xin, X.}
\newblock Titan m makes pixel 3 our most secure phone yet.
\newblock
  \url{https://blog.google/products/pixel/titan-m-makes-pixel-3-our-most-secure-phone-yet/},
  Oct. 2018.
\newblock Accessed: 13.11.2018.

\end{thebibliography}


\appendix





\end{document}